\documentclass{comnet}%

\usepackage{amsmath} %
\usepackage{amsthm} %
\usepackage{amssymb}
\usepackage{amsfonts}
\usepackage{dsfont}
\usepackage{graphicx} %
\usepackage{subfig} %
\usepackage{tikz,pgfplots}      %
\usetikzlibrary{arrows}
\usetikzlibrary{positioning}
\usepackage{xcolor} %
\pgfplotsset{compat=1.14}

\DeclareMathOperator{\degout}{\deg_{\text{out}}}
\DeclareMathOperator{\degin}{\deg_{\text{in}}}
\newcommand{\abs}[1]{\left|#1\right|}
\newcommand{\norm}[1]{\left\|#1\right\|}
\DeclareMathOperator{\Diag}{diag}

\DeclareMathOperator{\diag}{diag}

\DeclareMathOperator{\Real}{Re}

\newcommand{\de}{\text{d}}
\newcommand{\ddt}{\frac{\de}{\de t}}
\newcommand{\C}{\mathbb{C}}
\newcommand{\R}{\mathbb{R}}

\newcommand{\Z}{\mathbb{Z}}

\DeclareMathOperator{\sign}{sign}
\newcommand{\dd}{\,\text{d}}    %

\begin{document}

    \title{Nonlocal network dynamics via fractional graph Laplacians}

    \shorttitle{Nonlocal network dynamics} %
    \shortauthorlist{M. Benzi, D. Bertaccini, F. Durastante, I. Simunec} %

    \author{{%
            \sc Michele Benzi}$^*$\,,\\[2pt]
        Scuola Normale Superiore\\
        Piazza dei Cavalieri 7, 56126 Pisa, Italy\\
        {Corresponding
            author: {\email{michele.benzi@sns.it}}}\\[6pt]
        {\sc Daniele Bertaccini}\,, \\[2pt]
        Universit\`a di Roma Tor Vergata, Dipartimento di Matematica\\
        Via della Ricerca Scientifica 1, 00133 Roma, Italy
        \\Consiglio Nazionale delle Ricerche, Istituto per le Applicazioni del Calcolo ``M. Picone'', Roma, Italy\\
        {\email{bertaccini@mat.uniroma2.it}}\\[6pt]
        {\sc Fabio Durastante}\,, \\[2pt]
        Consiglio Nazionale delle Ricerche, Istituto per le Applicazioni del Calcolo ``M. Picone''\\
        Via Pietro Castellino 111, 80131 Napoli, Italy\\
        \email{f.durastante@na.iac.cnr.it}\\[6pt]
        {\sc and}\\[6pt]
        {%
            \sc Igor Simunec}\,,\\[2pt]
        Scuola Normale Superiore\\
        Piazza dei Cavalieri 7, 56126 Pisa, Italy\\
        {\email{igor.simunec@sns.it}}}

    \maketitle

    \begin{abstract}
        {We introduce nonlocal dynamics
            on directed networks through the construction of a
            fractional version of a nonsymmetric Laplacian for weighted directed
            graphs. Furthermore, we provide an analytic treatment of fractional dynamics
            for both directed and undirected graphs, showing the possibility of
            exploring the network employing random walks with jumps of
            arbitrary length. We also provide some examples of the applicability of
            the proposed dynamics, including consensus over multi-agent systems described
            by directed networks.}
        {network dynamics, nonlocal dynamics, superdiffusion, matrix functions, power law decay}
        \\
        2010 Math Subject Classification: 91D30, 60J20, 94C15
    \end{abstract}

    \section{Introduction}\label{sec:introduction}

    Systems made of highly interconnected units, where the connection
    stands for a kind of (possibly one-directional) interaction between the different
    nodes, are a ubiquitous modeling approach for several natural and man-made phenomena.
    Examples include social interactions in the real and digital world, gene regulatory networks,
    networks of chemical reactions, phone call networks, and many others. An efficient way for
    representing these complex interactions is through the use of graphs models. One of the main goals
    in this framework is to develop techniques and measures that are capable to characterize the topology
    of real networks, i.e., of graphs whose structure is irregular, complex and, possibly, evolving in time.
    A highly successful approach is to explore the network structure by means of random walks and
    other diffusive-type processes defined on the underlying graph.

    Here we investigate the behavior of certain nonlocal dynamical processes evolving on the network.
    In these models, a random
    walker on the network is not constrained to hop only from one node to adjacent nodes, but is allowed to
    perform long distance jumps, albeit with a lower probability. One can also phrase these processes as
    \textit{anomalous diffusion} phenomena, or \textit{superdiffusion}.

    Recently, two main approaches have been proposed to construct such nonlocal dynamics on graphs.
    The first one can be expressed in terms of multi hopper exploration strategies on the
    network~\cite{estradamultihopper,Estrada2017307,Estrada2018373}, leading to a probability
    distribution that permits the hopper to (occasionally) perform long distance jumps.
    Such  long range transitions, often referred to as \textit{L\'evy flights}, can also be described in the framework
    of the fractional calculus~\cite{randomwalkfractional}. Recent papers investigated some aspects
    of this phenomenon (in terms of anomalous diffusion) in the case of undirected graphs, making
    use of the (symmetric) fractional graph Laplacian and its normalized version; see~\cite{PhysRevE90032809}.

    The present paper has two main goals. One of them is to extend the notion of nonlocal dynamics to
    directed networks, investigating how the network structure
    affects the properties of a dynamical system evolving on it while accounting for the orientation of the connections.
    Similar to the approach in~\cite{PhysRevE90032809}, the
    method we propose can be formulated as the problem of evolving a system of ordinary differential equations
    in time using as coefficient matrix the fractional powers of a Laplacian of the underlying graph; an important
    difference, however, is that in the directed case the Laplacian matrix is nonsymmetric, hence the definition
    of fractional powers is more delicate than in the undirected (symmetric) case.
    This is done in Section~\ref{sec:fractional_laplacian_of_a_directed_graph}.
    Our second goal is related to the work presented in in~\cite{PhysRevE90032809} and consists in a rigorous
    analysis of the decay behavior  in the entries of the $\alpha$th power of the Laplacian matrix and its exponential.
    As we will see, in the undirected case (symmetric Laplacian) we can obtain very general results, applicable to
    virtually any network.  To complement the analysis, we show also that the fractional Laplacian of some simple infinite
    graphs induces a stable probability distribution with superdiffusive properties. Specifically,
    in Section~\ref{sec:notes_on_the_decay} we explore the decay properties of the transition probabilities
    of nonlocal random walks induced by the fractional Laplacian of undirected networks, and we offer some remarks on
    their possible extension to the directed case. Subsequently,
    in Section~\ref{sec:superdiffusive_processes} we analyze the superdiffusive behavior of the proposed dynamic
    on some simple infinite graphs (both directed and undirected), proving that it appears naturally as a stationary distribution
    for both the undirected and directed case by exploiting the techniques used in~\cite{Estrada2017307} for the $k$--path Laplacian;
    see Section~\ref{sec:random_walks_on_directed_graphs}.
In the directed case, the dynamics exhibit some similarities but
also interesting differences with respect to to the undirected case.
Finally, we consider two applications to real world directed
networks in Section~\ref{sec:applications}.

    \subsection{Preliminaries and notation on graphs}
    We recall here some basic notions on graphs that will be used in the
    following discussions. A \textit{directed graph}, or \textit{digraph}, is a pair $G = (V,E)$, where
    $V = \{v_1,\ldots,v_n\}$ is a set of nodes (or vertices), and $E~\subseteq~V \times V$ is a set of ordered pairs of nodes
    called edges. {We define on $V$ the binary relation $v_i \sim v_j$ if $(v_i,v_j) \in E$, or $(v_j,v_i) \in E$}. A \textit{weighted} directed graph $G~=~(V,E,W)$ is then
    obtained by considering a weight matrix $W$ with nonnegative entries
    $(W)_{i,j}=w_{i,j} \geq 0$ and such that  $w_{i,j} > 0$ if and only
    if $(v_i,v_j)$ is an edge of $G$. If all the nonzero weights have
    value $1$ we omit the weighted specification. For every node $v \in
    V$, the \textit{degree} $\deg(v)$ of $v$ is the number of edges leaving or
    entering~$v$ taking into account their weights,
    \begin{equation}\label{eq:definition_degree}
    d_i = \deg(v_i) = \sum_{j\,:\, { v_i \sim v_j} }w_{i,j}.
    \end{equation}
    A vertex is \textit{isolated} if its degree is zero.

    The degree matrix $D$ is then the diagonal matrix whose entries are given by the degrees of the nodes,~i.e.,
    \begin{equation}\label{eq:definition_degree_matrix}
    \begin{split}
    D = & \operatorname{diag}(\deg(v_1),\ldots,\deg(v_n)) =  \operatorname{diag}(d_1,\ldots,d_n).
    \end{split}
    \end{equation}
    In light of the fact that we want to consider dynamical processes on
    directed graphs, it is useful to separate the degrees also between
    the incoming and outgoing edges with respect to the node $v_i$,
    i.e., to consider the \textit{in--degrees} and \textit{out--degrees}
    \begin{equation*}
    d_i^{(\text{in})} = \degin(v_i) = \sum_{j\,:\, (v_j,v_i) \in E}w_{j,i}, \qquad d_i^{(\text{out})} = \degout(v_i) = \sum_{j\,:\, (v_i,v_j) \in E}w_{i,j},
    \end{equation*}
    together with the related diagonal matrices $D_{\text{in}} =  \operatorname{diag}(\degin(v_1),\ldots,\degin(v_n))  =  \operatorname{diag}(d^{(\text{in})}_1,\ldots,d^{(\text{in})}_n)$, and $D_{\text{out}} =  \operatorname{diag}(\degout(v_1),\ldots,\degout(v_n)) = \operatorname{diag}(d^{(\text{out})}_1,\ldots,d^{(\text{out})}_n)$. Moreover, we assume from now on that no vertex of the graph is isolated, and that all the graphs are loop-less, i.e.,
    that there is no edge going from a vertex to itself. Given a weighted directed graph $G = (V,E,W)$ with $V = \{v_1,\ldots,v_n\}$ {and $E=\{e_1,\ldots,e_m\}$}, the incidence matrix $B$ of $G$ is the $n\times
    m$ matrix whose entries $b_{i,j}$ are given~by
    \begin{equation}\label{def:weighted_digraph_incidence_matrix}
    b_{i,j} = \left\lbrace\begin{array}{ll}
    {+ \sqrt{w_{i,k}}}, & \text{ if } e_j = (v_i,v_k) \text{ for some }k,  \\
    {- \sqrt{w_{k,i}}}, & \text{ if } e_j = (v_k,v_i) \text{ for some }k,  \\
    0, & \text{ otherwise.}
    \end{array}{}\right.
    \end{equation}
    Observe that the choice of the sign in $B$ is purely conventional.

    If the ordering of the vertices in the edges in $E$ is not relevant,
    i.e., if each edge can be traversed both ways, we move from directed
    graphs to undirected graphs; that is, an undirected graph is a pair $G
    = (V,E)$, where $V=\{v_1,\ldots,v_n\}$ is a set of nodes or
    vertices, and $E \subseteq V\times V$ is a set of edges such that if $(v_i,v_j) \in E$, then
    $(v_j,v_i) \in E$ for all~$i,j$. A weighted undirected graph $G =
    (V,E,W)$ is then obtained by considering a (symmetric) weight matrix $W$ with
    nonnegative entries $(W)_{i,j}=w_{i,j} \geq 0$ and such that
    $w_{i,j} > 0$ if and only if $(v_i,v_j)$ is an edge of $G$. If all
    the nonzero weights have value $1$ we omit the weighted
    specification. For any two
    nodes $u,v \in V$ in a graph $G = (V,E)$, a \textit{walk} from $u$ to $v$ is
    an ordered sequence of nodes $(v_0,v_1,\ldots,v_k)$ such that $v_0 =
    u$, $v_k = v$, and $(v_i,v_{i+1}) \in E$ for all~$i=0,\ldots,k-1$.
    The integer $k$ is the length of the walk. The walk is closed if the
    initial and terminal nodes coincide, i.e., $u=v$. A \textit{cycle} in a graph
    is a nonempty closed walk in which the only repeated vertices are the first and last. An undirected graph $G$ is
    \textit{connected} if for any two distinct nodes $u,v \in V$, there is a walk
    between $u$ and $v$. A directed graph $G$ is \textit{strongly connected} if for
    any two distinct nodes $u,v \in V$, there is a directed walk from $u$ to $v$.
    For both a directed and an undirected graph $G$ we introduce the adjacency matrix $A$ as the $n \times n$ matrix with elements
    \begin{equation*}
    (A)_{i,j} = a_{i,j} = \left\lbrace\begin{array}{cc}
    1, & \text{ if }(v_i,v_j) \in E,  \\
    0, & \text{ otherwise}.
    \end{array}\right.
    \end{equation*}
    Observe that the adjacency matrix $A$ of an undirected graph $G$ is always
    symmetric. In particular, if $G = (V,E)$ is a graph,
    given two nodes $u,v \in V$, we say that $u$ is adjacent to $v$ and
    write $u \sim v$, if $(u,v) \in E$. The above binary relation is
    symmetric if $G$ is an undirected graph, while in general it is not for a
    directed graph. Note that for an unweighted graph, $W=A$.

    \subsubsection{Graph Laplacian}
Next, we recall the definition of the Laplacian matrix for an
undirected graph and then discuss an extension of it we will use in
the case of directed graphs.
    \begin{definition}[Graph Laplacian]\label{def:graph_laplacian}
        Let $G = (V,E)$ be a weighted undirected graph with weight matrix $W$,
        weighted degree matrix $D$ and weighted incidence matrix $B$. Then the graph Laplacian $L$ of $G$ is
        \[L = D - W = BB^T.\]
        The \textit{normalized random walk} version of the graph Laplacian is
        \[D^{-1}L = I - D^{-1}W = D^{-1}BB^T,\]
        where $I$ is the identity matrix. Observe that $D^{-1}W$ is a row--stochastic matrix, i.e. it is nonnegative with row sums equal to 1. The \textit{normalized symmetric} version is
        \[D^{-\frac{1}{2}} L D^{-\frac{1}{2}} = I - D^{-\frac{1}{2}} W D^{-\frac{1}{2}}.\]
        If $G$ is unweighted then $W=A$ in the above definitions. Here we assume that every vertex has nonzero degree.
    \end{definition}

    In the case of a directed graph the situation is more intricate
    since many nonequivalent definitions of the Laplacian
    exist. We can
    easily define, mimicking Definition~\ref{def:graph_laplacian}, the
    nonnormalized version with respect to the in-- and out--degrees, in
    both the weighted and unweighted case.
    \begin{definition}[Directed graph Laplacian]\label{def:digraph_laplacian}
        Let $G = (V,E,W)$ be a weighted directed graph, with degree matrices $D_{\text{out}}$ and
        $D_{\text{in}}$ The nonnormalized directed graph Laplacian
        $L_{\text{out}}$ and $L_{\text{in}}$ of $G$ are
        \[L_{\text{out}} = D_{\text{out}} - W, \qquad L_{\text{in}} = D_{\text{in}} - W.\]
    \end{definition}
    To define the normalized versions, we need to
    invert either the $D_{\text{in}}$ or the $D_{\text{out}}$ matrices,
    but the absence of isolated vertices is no longer sufficient to ensure
    this, since there could be a node with only outgoing or ingoing edges. A
    first way of overcoming this issue could be to impose that every
    vertex has at least one outgoing and one incoming edges, which is
    rather restrictive. Otherwise, we could restrict our attention to the
    set of nodes having an out--degree or in--degree different from
    zero, as in~\cite{MR2915277}. Another approach, that avoids reducing the size of the graph, is instead to mimic the recipe for the PageRank algorithm~\cite{page1999pagerank} and replace any diagonal zeros in $D_{\text{in}}$, respectively $D_{\text{out}}$, with ones, while replacing the corresponding (zero) column, respectively row, of $W$ with the vector with entries~$1/n$.

    The last approach we briefly mention is the one presented in~\cite{MR2135772}.
    In this case a symmetric Laplacian is constructed also for a directed
    graph. However, it is easy to see that this kind of approach may return the same
    Laplacian matrix for nonisomorphic graphs. This also happens if we define a symmetric digraph Laplacian by using the
    incidence matrix $B$ of
    Definition~\ref{def:weighted_digraph_incidence_matrix} and construct $L = BB^T$
    as in Definition~\ref{def:graph_laplacian}. In the rest of the paper
    we focus mainly on the nonsymmetric Laplacian $L_{\text{out}}$ and
    its normalized version.

    \section{Fractional Laplacians of a directed graph}\label{sec:fractional_laplacian_of_a_directed_graph}
    To justify the use of a fractional Laplacian for exploring the structure of the
    network, let us first consider a simple diffusion problem in the case in which $G$ is an undirected graph. Let $u :
    V \rightarrow \mathbb{R}$ describe a ``heat'' distribution on the
    nodes of the graph with heat diffusivity $\kappa$. We can express
    the variation of heat in the nodes as
    \begin{equation*}
    \begin{split}
    \displaystyle \ddt u(t) = & - \kappa\sum_{j\,:\, (v_j,v_i) \in E} (u_i - u_j)
    =  -\kappa \left(u_i \sum_{j\,:\, (v_j,v_i) \in E} 1 - \sum_{j\,:\, (v_j,v_i) \in E} u_j\right) \\
    = & - \kappa\left(u_i \deg(v_i) - \sum_{j\,:\, (v_j,v_i) \in E} u_j\right)
    =  - \kappa\sum_{j\,:\, (v_j,v_i) \in E}\left(\delta_{i,j} \deg(v_i) -  1\right)u_j \\
    = & - \kappa\sum_{j\,:\, (v_j,v_i) \in E} (L)_{i,j} u_j,
    \end{split}
    \end{equation*}
    which in matrix form reads
    \begin{equation}\label{eq:diffusion_formulation}
    \begin{split}
    \text{find }u\, :&\, [0,T] \longrightarrow \mathbb{R}^n\\ \text{ s.t. }& \left\lbrace\begin{array}{ll}
    \displaystyle \ddt u(t) = - \kappa L u(t), &t \in (0,T],  \\
    u(0) = u_0, &\text{prescribed,}
    \end{array}\right.
    \end{split}
    \end{equation}
    where now $L$ is the unweighted Laplacian from
    Definition~\ref{def:graph_laplacian}. Since $L$ is a symmetric positive semidefinite
    matrix, one can apply the process of ``fractionalization'' considered in~\cite{ilic2005numerical,ilic2006numerical} for the
    continuous Laplace operator. Consider the spectral decomposition of the Laplacian matrix,
    \begin{equation*}
    L = U \Lambda U^T, \qquad U^TU = I, \qquad \Lambda = \operatorname{diag}(\lambda_1, \dots, \lambda_n).
    \end{equation*}
    Following~\cite{PhysRevE90032809}, we define the \textit{fractional graph Laplacian} as
    \begin{equation}\label{eq:fractional_laplacian_of_a_graph}
    \begin{split}
    &L^\alpha = U \Lambda^\alpha U^T, \qquad U^TU = I, \\
    &\Lambda^\alpha = \operatorname{diag}(\lambda_1^\alpha,\ldots,\lambda_n^\alpha), \qquad
    \alpha \in (0,1].
    \end{split}
    \end{equation}
    Note that the fractional powers $\lambda_1^\alpha, \dots, \lambda_n^\alpha$ are well defined because the eigenvalues of the Laplacian matrix are nonnegative. This follows from the fact that the Laplacian is an $M$-matrix, see Definition~\ref{def:M-matrix}.

    The definition of the fractional graph Laplacian becomes significantly different when the case of the (nonnormalized) digraph Laplacian from
    Definition~\ref{def:digraph_laplacian} is considered. In general, this operator is non normal, and thus we cannot define the fractional
    power as
    in~\eqref{eq:fractional_laplacian_of_a_graph}. Therefore, we need to define the $\alpha$th
    power of a non normal matrix.

    Without loss of generality, we focus the analysis on the
    out--degree Laplacian $L_{\text{out}} = D_{\text{out}} - A$ since it
    remains essentially the same in the case of the in--degree
    Laplacian. We first recall a suitable definition for the matrix
    function $f(A)$ for a generic matrix $A$, that extends the one based on the
    diagonalization in~\eqref{eq:fractional_laplacian_of_a_graph}.
    This definition can be stated in terms of the Jordan canonical form of the matrix~\cite[Section~1.2.2]{MR2396439}.

        {We recall that a}ny matrix $A \in \mathbb{C}^{n \times n}$ can be expressed in Jordan canonical
        form as
        \begin{equation}\label{eq:def:jordan_canonical_form}
        Z^{-1} A Z = J = \operatorname{diag}(J_1,\ldots,J_p), \quad \text{ for }
        J_k = J_k(\lambda_k) = \begin{bmatrix}
        \lambda_k & 1 \\
        & \lambda_k & \ddots \\
        & & \ddots & 1 \\
        & & & \lambda_k
        \end{bmatrix} \in \C^{m_k \times m_k},
        \end{equation}
        where $Z$ is nonsingular and $m_1 + m_2 + \ldots + m_p = n$. If each block in which the eigenvalue $\lambda_k$ appears is of size 1 then $\lambda_k$ is said to be a \textit{semisimple} eigenvalue.

    Let us denote by $\lambda_1,\ldots,\lambda_s$ the distinct
    eigenvalues of $A$, and by $n_i$ the order of the largest Jordan block in
    which the $\lambda_i$ appears, i.e., the \textit{index} of the eigenvalue
    $\lambda_i$.  We have the following definition.
    \begin{definition}\label{def:defined_on_the_spectrum}
        The function $f$ is \textit{defined on the spectrum of} $A$ if the values
        \begin{equation*}
        f^{(j)}(\lambda_i), \qquad j=0,1,\ldots,n_i-1, \quad i = 1,\ldots s,
        \end{equation*}
        exist, where $f^{(j)}$ denotes the $j$th derivative of $f$, with $f^{(0)} = f$.
    \end{definition}
    We can define the matrix function $f(A)$ for a generic matrix $A$ by
    using the Jordan canonical form, provided that the function $f$ is defined on the spectrum of $A$.
    \begin{definition}\label{def:matrix_function_as_jordan_form}
        Lef $f$ be defined on the spectrum of $A \in \mathbb{C}^{n \times
            n}$, which is represented in Jordan canonical form as in~\eqref{eq:def:jordan_canonical_form}. Then,
        \begin{equation*}
        f(A) = Z f(J) Z^{-1} = Z \operatorname{diag}(f(J_1),\ldots,f(J_p))Z^{-1},
        \end{equation*}
        where
        \begin{equation*}
        f(J_k) = \begin{bmatrix}
        f(\lambda_k) & f'(\lambda_k) & \ldots & \frac{f^{(m_k-1)}(\lambda_k)}{(m_k-1)!} \\
        & f(\lambda_k) & \ddots & \vdots \\
        & & \ddots & f'(\lambda_k) \\
        & & & f(\lambda_k)
        \end{bmatrix}.
        \end{equation*}
        Moreover, let $f$ be a multivalued function and suppose some eigenvalues occur in more than one Jordan block. If the same choice of branch of $f$ is made in each block, then we say that $f(A)$ is a \textit{primary matrix function}. In this paper we only consider primary matrix functions.
    \end{definition}
    Note that in the real symmetric case, given
    in~\eqref{eq:fractional_laplacian_of_a_graph}, the Jordan canonical
    form reduces to the diagonalization of the matrix.

    In order to ensure that $f(L_{\text{out}}) = L_{\text{out}}^\alpha$, $\alpha
    \in (0,1]$ is well defined, we need to check first that $f(x) =
    x^\alpha$ is defined on the spectrum of $L_{\text{out}}$
    (Definition~\ref{def:defined_on_the_spectrum}). In the following discussion, by $f(x) = x^\alpha$
    we refer to the branch with a cut on the negative real line, i.e. if $x = \rho e^{i\theta}$ with $\rho>0$
    and $\theta \in (-\pi,\pi)$, then $x^\alpha = \rho^\alpha e^{i\alpha\theta}$.

    This function is defined on the spectrum of the Laplacian because, as in the
    symmetric case, the matrix $L_{\text{out}}$ is a singular
    $M$-matrix, {with $0$ as a semisimple eigenvalue}.
    \begin{definition}[$M$-matrix,~\cite{MR1298430}]
        \label{def:M-matrix}
        A matrix $A \in \mathbb{R}^{n \times n}$ is an \textit{$M$-matrix}
        if $A = s I - B$ for some nonnegative matrix $B$, where $s \geq \rho(B)$, the
        spectral radius of $B$. It is a \textit{singular} $M$-matrix if $s = \rho(B)$.
    \end{definition}
    Note that the real part of a nonzero eigenvalue of a singular $M$-matrix is positive, and that the $M$-matrices form a \textit{closed subset} $\mathcal{M}$ of the vector space of real
    matrices $\mathbb{M}_{n}$; %
    we refer to \cite{MR1298430} for further information regarding
    these matrices, including the following basic result, where we denote by $\boldsymbol{0}$ the vector of all zeros and by $\boldsymbol{1}$ the vector of all ones.
    \begin{proposition}[Properties of $L_{\text{out}}$]\label{pro:properties_of_outdegree_laplacian}\leavevmode
        \begin{itemize}
            \item $L_{\text{out}}$ is a singular $M$-matrix,
            \item $L_{\text{out}} \boldsymbol{1} = \boldsymbol{0}$,
            \item $0$ is a semisimple eigenvalue of $L_{\text{out}}$.
        \end{itemize}
    \end{proposition}
    As a consequence we have the following Theorem.
    \begin{theorem}\label{pro:wellposednessofalfapower}
        Given a weighted graph $G = (V,E,W)$ and its Laplacian
        with respect to the out degree $L_{\text{out}}$
        (Definition~\ref{def:digraph_laplacian}), the function $f(x) = x^\alpha$ is defined on the spectrum of
        $L_{\text{out}}$ and induces a matrix function for all $\alpha \in (0,1]$.
    \end{theorem}

    \begin{proof}[Proof of Theorem~\ref{pro:wellposednessofalfapower}]
        By Proposition~\ref{pro:properties_of_outdegree_laplacian} we know
        that~$0$ is a semisimple eigenvalue of $L_{\text{out}}$, then all
        the Jordan blocks related to the eigenvalue $\lambda_1 = 0$ have
        size 1 and $f(\lambda_1) = f(0)$ exists. Since $L_{\text{out}}$ is
        a singular $M$-matrix, $\operatorname{Re}(\lambda_k) > 0$,
        for all~$\lambda_k \neq 0$, and $f^{(j)}(\lambda_k)$ exist
        for all~$j$. Thus, by
        Definition~\ref{def:matrix_function_as_jordan_form}, $f$ is defined
        on the spectrum of~$L_{\text{out}}$. Moreover, let $\lambda$ be any
        nonzero eigenvalue of~$L_{\text{out}}$. Then, $\lambda = \rho e^{i
            \theta}$ with $\theta \in (-\frac{\pi}{2},\frac{\pi}{2})$ since $\operatorname{Re}(\lambda_k) > 0$ and
        thus we can define $\lambda^\alpha = \rho^\alpha e^{i \alpha \theta}$ with
        $\alpha\theta \in (-\frac{\pi}{2},\frac{\pi}{2})$. Therefore, we can always select the
        branch of $x^\alpha$ preserving the positivity of the real part
        of the eigenvalues, thus ensuring the choice of a primary matrix
        function.
    \end{proof}

    Under the same hypothesis of
    Theorem~\ref{pro:wellposednessofalfapower} we can say more
    about the structure of $L_{\text{out}}^\alpha$. Indeed, this is a result
    that is already known for the special case of the matrix $p$th root.
    \begin{theorem}[\cite{MR2599831}]\label{thm:pth_root_is_m_matrix}
        If $A$ is a singular $M$-matrix {with 0 as a semisimple eigenvalue},
        then there exists a determination of $A^{1/p}$ for every $p \in \mathbb{N}$ that is a singular $M$-matrix.
    \end{theorem}
    Similarly, we get the following useful result. %
    \begin{theorem}\label{thm:alphapowersingularmmatrix}
        If $A$ is a singular $M$-matrix {with 0 as a semisimple eigenvalue},
        then there exists a determination of $A^{\alpha}$ for every $\alpha \in (0,1]$ that is a singular $M$-matrix.
    \end{theorem}

    \begin{proof}[Proof of Theorem~\ref{thm:alphapowersingularmmatrix}]
        Let $A(\varepsilon) = A + \varepsilon I$, then $A(\varepsilon)$ is a
        nonsingular M--matrix and so is $A(\varepsilon)^{\alpha}$~\cite[Corollary~3.7]{MR700883}. By looking at the Jordan
        canonical forms of the matrices $A(\varepsilon)^{\alpha}$ and
        $A^\alpha$ (Theorem~\ref{pro:wellposednessofalfapower}), we get
        $A(\varepsilon)^{\alpha} \rightarrow A^\alpha$ for $\varepsilon
        \rightarrow 0$. {Clearly, $A^\alpha$ is singular, and since the $M$-matrices from a closed subset of $\mathbb{M}_n$,
        we conclude that $A^{\alpha}$ is a singular $M$-matrix.}
    \end{proof}

    Moreover, note that the matrix produced in this way is a primary
    matrix function since we selected the same branch of the $f(x) =
    x^\alpha$ for every matrix of the sequence.

    \section{Decay bounds for the entries of fractional Laplacians}
    \label{sec:notes_on_the_decay}
    Quantitative estimates for the entries of fractional powers of the graph Laplacian yield valuable information on
    the transition probabilities of various types of random walks on the underlying graph. In this section we show
    how to obtain useful bounds for these quantities using general results on functions of matrices, at least in the case of
    undirected networks. We also comment on the difficulties one encounters when trying to extend such results to
    the case of directed graphs.

    \subsection{Undirected networks}
    \label{sec:notes_on_the_decay_undirected}

    First of all, we show that the fractional Laplacian $L^\alpha$ is related to a row--stochastic matrix, which can be used to define a fractional random walk on the graph, similarly to the Laplacian $L$ (see Definition~\ref{def:graph_laplacian}).
    \begin{lemma}
        \label{lemma:fractional random walk} For $\alpha \in (0,1)$, the
        matrix
        \begin{equation*}
        P^{(\alpha)} = I - \bar{L}^{(\alpha)}, \qquad \text{where }\bar{L}^{(\alpha)} =
        \diag(L^{\alpha})^{-1} L^{\alpha},
        \end{equation*}
        is a row--stochastic matrix.
    \end{lemma}
    \begin{proof}
        We start by noting that all the diagonal entries of $L^\alpha$ are positive, so we
        have $\diag(\bar{L}^{(\alpha)}) = I$ and thus the diagonal entries of $P^{(\alpha)}$
        are zero. This can be seen by explicitly writing the eigendecomposition of $L$ and $L^\alpha$.

        Given that $L \boldsymbol{1} = \boldsymbol{0}$, we also have $\bar{L}^{(\alpha)} \boldsymbol{1} = \boldsymbol{0}$, so it is sufficient to show that $P^{(\alpha)} \ge 0$.

        We can write $ L = D - A = \rho I - B$, where $\rho = \max_i{d_i}$ and $B$
        is obtained from $A$ by increasing the diagonal entries so that its
        row sums are all equal to $\rho$. Therefore,
        \begin{equation*}
        L^\alpha = \rho^\alpha (I - {\textstyle\frac{1}{\rho}}B)^\alpha =
        \rho^\alpha\sum_{k = 0}^{\infty} \binom{\alpha}{k}(-1)^k
        \frac{1}{\rho^k}B^k, \qquad \text{where}\quad \binom{\alpha}{k} =
        \frac{\alpha(\alpha-1)\cdots(\alpha-k+1)}{k!}.
        \end{equation*}
        We have $\norm{{\rho^{-1}}B}_{\infty} = 1$, so the series {of infinity norms is bounded from above by $\displaystyle\sum_{k=0}^\infty \binom{\alpha}{k}(-1)^k$, which }is absolutely
        convergent since $\alpha > 0$; {this implies that the above matrix series for $L^\alpha$ is convergent}.
        Moreover, since $\alpha \in (0,1)$, we have
        that $\displaystyle\binom{\alpha}{k} > 0$ when $k$ is odd, and
        $\displaystyle\binom{\alpha}{k} < 0$ when $k$ is even. So all terms
        with $k \ge 1$ of the sum for $L^\alpha = (I-\frac{1}{\rho}B)^\alpha$ are nonpositive, since
        $B$ has non negative entries.
        We conclude by observing that $P^{(\alpha)} = I - \diag(L^\alpha)^{-1}L^\alpha$, so all the offdiagonal entries of $P^{(\alpha)}$ are nonnegative.

    \end{proof}

    We can interpret the random walk with transition matrix $P^{(\alpha)}$ as the
    one induced by the weighted undirected graph with adjacency matrix
    $$
    A_\alpha = \Diag(L^\alpha) P^{(\alpha)} = \Diag(L^\alpha) - L^\alpha.
    $$
    The entries of the vector $d_\alpha = \diag(L^\alpha)$ are the fractional
    degrees associated to $A_\alpha$, and they give us the
    stationary distribution of the random walk as in the standard case:
    \begin{equation*}
    \pi_\alpha^T P^{(\alpha)} = \pi_\alpha^T \quad \iff \quad \pi_\alpha = \frac{1} {\boldsymbol{1}^T
        d_\alpha} d_\alpha.
    \end{equation*}
    By analogy with the non fractional normalized Laplacian, we can use
    $\bar{L}^{(\alpha)}$ to define a continuous time random walk that solves the differential equation
    \begin{equation*}
    \begin{cases}
    \displaystyle \ddt u(t) = - \bar{L}^{(\alpha)} u(t), \\
    u(0) = u_0,
    \end{cases}
    \end{equation*}
    where $u_0$ is a given initial probability vector. The solution is given explicitly by
    \begin{equation*}
    u(t) = e^{-t\bar{L}^{(\alpha)}} u_0 =
    e^{-t\Diag(L^\alpha)^{-1}L^\alpha} u_0,
    \end{equation*}
    and is a probability distribution for every $t > 0$ whenever $u_0$ is, i.e., the entries of $u(t)$
    are between $0$ and $1$, and they sum up to $1$.

    Even if the graph Laplacian $L$ is sparse, its fractional powers
    $L^\alpha$, $\alpha \in (0,1)$ are usually full matrices.
    However, functions of sparse matrices can have entries that
    decay rapidly in magnitude far from the nonzero pattern of the original
    matrix~\cite{MR2455657,MR3391978,MR1812534}. In
    particular, for a function $f$ that is analytic on the convex hull
    of the spectrum of the symmetric matrix $A$, the decay in the
    entries of $f(A)$ is exponential, or superexponential if $f$ is an entire function.
    On the other hand, if $f$ is not analytic, the decay can be slower; the lower the
    regularity of $f$, the slower the decay.

    In the cases we consider, the functions $f(x) = x^\alpha$ and $g(x) = e^{- t x^\alpha}$, with $\alpha \in(0,1)$ are
    not differentiable in $x=0$, and the Laplacian matrix always has a (semisimple) eigenvalue at zero, given that $L \boldsymbol{1} = \boldsymbol{0}$.
    Hence, the exponential decay results for functions that are analytic on the spectrum of $L$ do not apply, and indeed numerically one observes
    much slower decay.  As it turns out, we can show that a power law decay occurs, using a well known approximation theorem for
    continuous functions defined on a compact interval, as we shall see below.

    The decay in the entries of the fractional Laplacian motivates the use of this matrix to model long-range diffusion and random walks on the graph.
    Indeed, the locality effect in the standard case derives from the superexponential decay of the entries of the related matrix function. When
    the fractional power of the Laplacian is used, the decay of the transition probabilities assumes a power law decay, hence the probability of performing
    a long jump is greatly increased with respect to the standard (classical diffusion) case, where these long range transitions are essentially impossible.

    \begin{theorem}[{Jackson's Theorem \cite[Theorem 43]{approx}}]
        \label{thm:Jackson}
        Let $f: [a,b] \to \R $ be a function with
        modulus of continuity $\omega$. Then, for any $n \ge 1$, the best approximation error
        $E_n(f)$ that can be obtained with polynomials of degree $\le n$ satisfies
        \begin{equation*}
        E_n(f) := \min_{\deg p_n \le n} \norm{f-p_n}_{\infty}\le c
        \omega\left(\frac{b-a}{2n}\right),
        \end{equation*}
        where $c = 1 + \pi^2/2$ is a constant independent
        of $n$ and of $f$.
    \end{theorem}
    We recall that the graph $G_M = (V_M, E_M)$ induced by a matrix $M \in \C^{n \times n}$ is the graph with nodes $V_M = \{v_1, \dots, v_n\}$ and edges $E_M = \{ (v_i, v_j) : M_{ij} \ne 0 \}$.
    \begin{proposition}
        \label{prop: general decay}
        Let $M$ be a symmetric matrix with spectrum $\sigma(M) \subset [a,b]$. Denote by $d(i,j)$ the distance between $i$ and $j$ in the graph induced by $M$, i.e. the length of the shortest path connecting nodes $i$ and $j$.
        Let $f:[a,b] \to \R$ be a function with modulus of continuity $\omega$. Then the following holds:
        \begin{equation*}
        \abs{f(M)_{ij}} \le  c \cdot \omega\left(\frac{b-a}{2}{[} d(i,j)-1 {]}^{-1}\right), \qquad d(i,j) \ge 2,
        \end{equation*}
        where $c = 1 + \pi^2/2$ is the constant from Jackson's Theorem \ref{thm:Jackson}.
    \end{proposition}
    \begin{proof}
        Note first that $f$ is defined on the spectrum of $M$, since $M$ is symmetric and thus diagonalizable. In particular, we have $M = Q \Lambda Q^T$, with $Q$ orthogonal and $\Lambda = \diag(\lambda_1, \dots, \lambda_n)$. Then, for any polynomial $p$ we have
        \begin{align*}
        \norm{f(M) - p(M)}_2 &= \norm{Q f(\Lambda) Q^T - Q p(\Lambda) Q^T}_2 \\
        &= \norm{f(\Lambda) - p(\Lambda)}_2 \\
        &= \norm{f(\lambda) - p(\lambda)}_{\infty, \sigma(M)} \le \norm{f(\lambda) - p(\lambda)}_{\infty,[a,b]}\,,
        \end{align*}
        where we have used basic properties of matrix functions and the invariance of the $2$-norm under orthogonal transformations.
        By Jackson's Theorem \ref{thm:Jackson}, we then have that for all $m \ge 1$ there exists a polynomial $p_m$ with $\deg p_m \le m$ such that
        \begin{equation}
        \label{eqn: proof 3}
        \norm{f(M) - p_m(M)}_2 \le \norm{f - p_m}_{\infty, [a,b]} \le c \cdot \omega\left(\frac{b-a}{2m}\right).
        \end{equation}
        Now, let us fix $i$ and $j \in \{ 1, \dots, n\}$. If $d(i,j) = m+1$, it is easy to see that all powers of $M$ up to the $m$-th have a zero entry in position $(i,j)$. Therefore $f(M)_{ij} = f(M)_{ij} - p_{m}(M)_{ij}$, and we obtain
        \begin{equation*}
        \abs{f(M)_{ij}} \le \norm{f(M) - p_m(M)}_2 \le c \cdot {\omega\left(\frac{b-a}{2m}\right)} = c \cdot \omega\left(\frac{b-a}{2}{[}d(i,j)-1{]}^{-1}\right).
        \end{equation*}
    \end{proof}

    \begin{remark}
        The result of Proposition~\ref{prop: general decay} only provides information for pairs of nodes that are at least a distance of $2$ apart. This is enough for our purposes, since we are mainly interested in sparse graphs, and in the behavior of transition probabilities for nodes that are far from each other.
    \end{remark}

    We can use the result of Propositon~\ref{prop: general decay} to obtain bounds on the entries of the fractional Laplacian of an undirected graph.
    \begin{corollary}
        \label{cor: decay for L^a}
        Let $L$ be the Laplacian of an undirected graph, $\alpha \in (0,1)$ and $t > 0$. Then, if $d(i,j) \ge~2$, the following inequalities hold:
        \begin{equation}
        \label{eqn: decay for L^a}
        \begin{aligned}
        \abs{(L^\alpha)_{ij}} &\le c \, \frac{\rho(L)^\alpha}{2^\alpha} \cdot {[}d(i,j) - 1{]}^{-\alpha}, \\
        \abs{\exp(-tL^\alpha)_{ij}} &\le c \cdot \left[1 - \exp\left({\textstyle -t \frac{\rho(L)^\alpha}{2^\alpha} {[}d(i,j)-1{]}^{-\alpha}}\right)\right]  \le ct \, \frac{\rho(L)^\alpha}{2^\alpha}\cdot {[}d(i,j)-1{]}^{-\alpha},
        \end{aligned}
        \end{equation}
        with $c = 1 + \pi^2/2$.
    \end{corollary}
    \begin{proof}
The first inequality follows immediately from Proposition~\ref{prop:
general decay}, because $f(x) = x^\alpha$ is $\alpha$--H\"older,
with modulus of continuity $\omega_f(x) = x^\alpha$.

The second set of inequalities also follows from Proposition
\ref{prop: general decay}, noticing that if $g(x) =
\exp(-tx^\alpha)$, for $x,y \ge 0$ it holds $g(x) - g(y) \le g(0) -
g(\abs{x-y})$, and thus the modulus of continuity of $g$ is
$\omega_g(x) = 1 - g(x)$; we conclude with the inequality
        \begin{equation*}
        e^{-x} \ge 1 - x, \qquad \forall x \geq 0.
        \end{equation*}
    \end{proof}

    \begin{corollary}
        \label{cor: decay for P_a}
        If $d(i,j) \ge~2$, the off-diagonal entries of $P^{(\alpha)} = I - \bar{L}^{(\alpha)}$ satisfy
        \begin{equation}
        \label{eqn: decay for P_a}
        \abs{(P^{(\alpha)})_{ij}} \le  c \frac{\rho(L)}{2^\alpha\abs{L_{ii}}} \cdot \abs{d(i,j) - 1}^{-\alpha}, \qquad \text{with } c = 1 + \pi^2/2.
        \end{equation}
    \end{corollary}
    \begin{proof}
It is sufficient to obtain a lower bound for the diagonal entries of
$L^\alpha$ and then use Corollary~\ref{cor: decay for L^a}.

        For $\lambda \in \sigma(L)$, we have $\lambda^\alpha \ge \rho(L)^{\alpha-1} \lambda$; by using this fact and the spectral decomposition $L^\alpha = \sum_{j = 1}^{n} \lambda_j^\alpha
        q_j q_j^T$, we get $(L^\alpha)_{ii} \ge \rho(L)^{\alpha-1} L_{ii}$.

    \end{proof}

We conclude this part with an example useful to illustrate the decay
of the entries of the fractional Laplacian.

    \begin{example}
We consider (the largest connected component of) the undirected
graph \texttt{DC} from the
collection~\texttt{users.diag.uniroma1.it/challenge9/data/tiger/},
which represents the road network of the city of Washington, DC.
Having fixed a node $i_0$ near to the center of the geographic
coordinates associated to the nodes of the network, we compare the
entries $(L^\alpha)_{k, i_0}$ with the distances $d(i_0,
k)^{-\alpha}$ for $\alpha  = 0.5$, for all $k$. The results,
summarized in Figure~\ref{fig:DC_decay}, closely match the behavior
proved in Corollary~\ref{cor: decay for L^a}.

        \begin{figure}[htbp]
            \centering
            \includegraphics[width=\columnwidth]{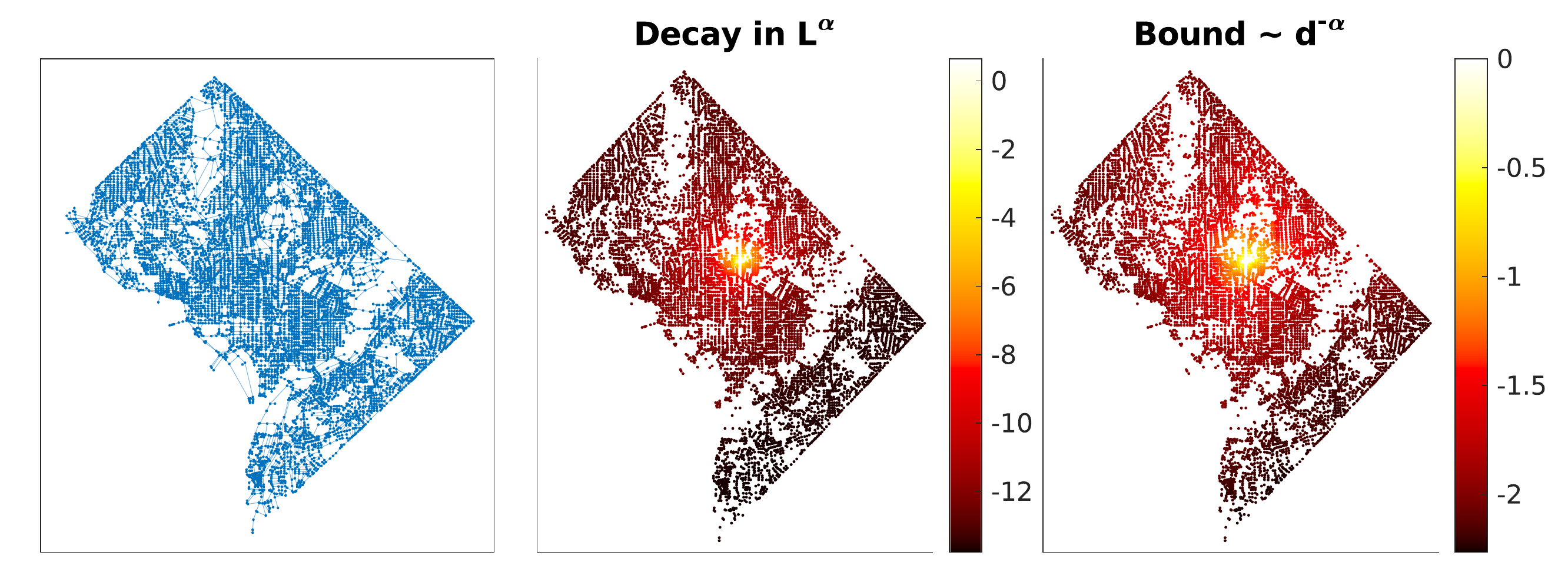}

            \caption{(Color online) Comparison between the entries of $L^\alpha$ and the distances
            between nodes in the graph \texttt{DC}.
Left panel: largest connected component of the graph, with $ n =
9522$ nodes. Central panel: decay in the entries $(L^\alpha)_{k,
i_0}$, for $k = 1, \dots, n$. Right panel: distances $d(i_0,
k)^{-\alpha}$, for $k = 1, \dots, n$. We used $\alpha = 0.5$ and
computed $L^\alpha$ via its
eigendecomposition~\eqref{eq:fractional_laplacian_of_a_graph}; the
scale for the colors is logarithmic.}
            \label{fig:DC_decay}
        \end{figure}
    \end{example}

    \subsection{Directed networks}
    Numerical evidence shows that the decay behavior in the entries of fractional Laplacians is not
    limited to the undirected case, but it can also be observed in directed networks; see
    Section \ref{sec:applications}. However, a generalization to the directed case of the results in Corollary~\ref{cor: decay for L^a} is not straightforward.

    If $A$ is a nonnormal matrix and $f$ is analytic on an open set containing the \textit{numerical range} $W(A)$ of $A$,
    \begin{equation*}
    W(A) = \left\lbrace \frac{\mathbf{x}^H A \mathbf{x}}{\mathbf{x}^H\mathbf{x}} \,:\, \mathbf{x} \in \mathbb{C}^n, \; \mathbf{x} \neq \mathbf{0}  \right\rbrace = \left\lbrace \mathbf{x}^H A \mathbf{x} \,:\, \mathbf{x} \in \mathbb{C}^n, \; \|\mathbf{x}\|_2 = 1 \right\rbrace,
    \end{equation*}
    it is shown in \cite{BB14} that one
    can bound the entries of $f(A)$ using the following result of Crouzeix~\cite{Crouzeix04, Crouzeix07}:
    \begin{equation}\label{eq:crouzeix-magic}
    \exists\, C  \text{ such that }\| f(A) \|_2 \leq C \sup_{w \in W(A)}
    |f(w)|,
    \end{equation}
    where $C$ is a universal constant independent of both $f$ and $A$; currently, the best known value for $C$ is $1+\sqrt{2}$, and it is
    conjectured to be 2.
    Unfortunately, (\ref{eq:crouzeix-magic}) cannot be used in our case
    since $f(x) = x^{\alpha}$, $\alpha \in (0,1)$, is not analytic on
    the negative real axis, and it is easy to find directed graphs
    such that the numerical range $W(L_{\text{out}})$ of the out--degree Laplacian contains part of the
    negative real axis. Indeed,
    \cite[Theorem~1.6.6]{MR2978290} states that if $\lambda$ is an
    eigenvalue of $L_{\text{out}}$ that lies on the boundary of
    $W(L_{\text{out}})$, then the eigenvector associated to it is
    orthogonal to all the {other} eigenvectors. Thus, if
    $W(L_{\text{out}}) \subseteq \mathbb{C}^+$ we have that $0 \in
    \lambda(L_{\text{out}}) \cap \partial W(L_{\text{out}})$, and then
    its eigenvector ${\bf 1}$ is orthogonal to all the other
    eigenvectors of $L_{\text{out}}$. Digraphs in which this does not
    happen are easy to find, and frequently encountered in applications.

    \begin{example} Consider the graph with adjacency matrix

    \vspace{0.1in}

    \begin{minipage}{0.425\textwidth}
        \centering
        \begin{tikzpicture}[>=stealth',shorten >=0pt,auto,node distance=1.8cm,thick,
        dot/.style={circle, fill, outer sep=1.5pt, inner sep=2pt}]
        \node[dot, label=180:{\small\bfseries 1}] (1) {};
        \node[dot, label=0:{\small\bfseries 2}] (2) [right=1.8cm of 1] {};
        \path (1) -- (2) node[midway] (aux) {};
        \node[dot, label=180:{\small\bfseries 3}] (3) [below=1.6cm of aux] {};
        \path[->]
        (1) edge (2)
        (2) edge [bend right] (3)
        (3) edge (1)
        (3) edge [bend right] (2);
        \end{tikzpicture}
    \end{minipage}%
    \begin{minipage}{0.575\textwidth}
        \begin{equation*}
        A = \begin{bmatrix}
        0 & 1 & 0 \\
        0 & 0 & 1 \\
        1 & 1 & 0 \\
        \end{bmatrix},
        \end{equation*}
        \null
        \par\xdef\tpd{\the\prevdepth}
    \end{minipage}

    \vspace{0.1in}

    \noindent whose out--Laplacian $L_{\text{out}}$ has the following
    field of values:

    \vspace{0.1in}

    \begin{minipage}{0.15\textwidth}
        \centering
        $W(L_{\text{out}}) = $
    \end{minipage}
    \begin{minipage}{0.3\textwidth}
        \centering
        \includegraphics[width=\linewidth,keepaspectratio=true]{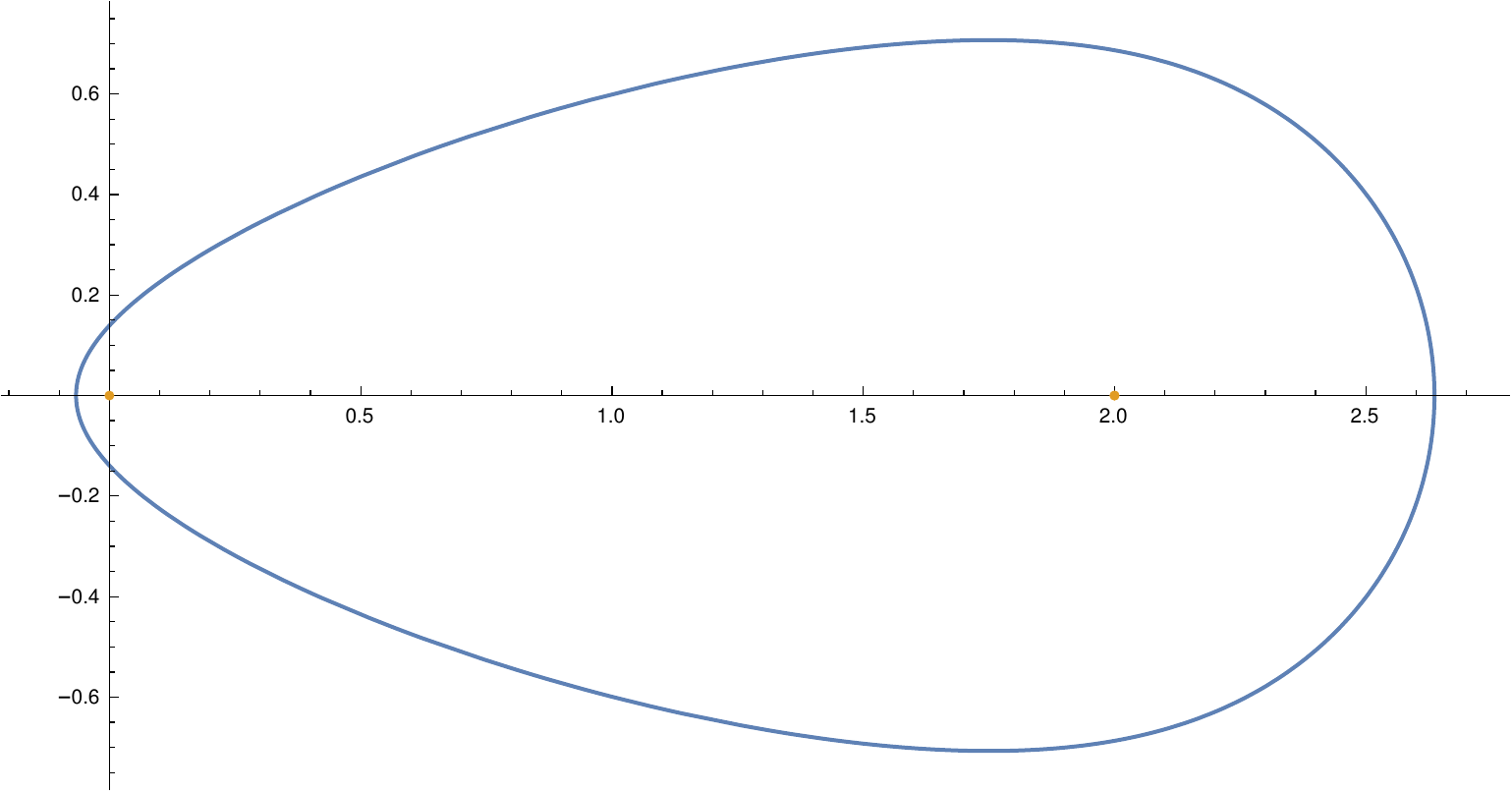}
    \end{minipage}%
    \begin{minipage}{0.55\textwidth}
        \begin{equation*}
        L_{\text{out}} = \begin{bmatrix}
        1 & -1 & 0 \\
        0 & 1 & -1 \\
        -1 & -1 & 2 \\
        \end{bmatrix},
        \end{equation*}
        \null
        \par\xdef\tpd{\the\prevdepth}
    \end{minipage}

    \vspace{0.75em}
    \noindent which includes part of the negative real axis and therefore the
    origin.
    \end{example}

    There are other possible general alternatives to
    (\ref{eq:crouzeix-magic}), which use extensions of the well known
    Dunford--Taylor integral representation of $f$, but attempts to
    bound the norm of terms like $|f(z)|=|z|^{\alpha}$, $\alpha\in
    (0,1)$ inside the contour integral cannot give a finite value and
    are not reported here.

    Nevertheless, there are special cases for which a reasonable bound can be
    provided. First, if the Laplacian matrix is diagonalizable and we can give a
    bound for the spectral condition number of an eigenvector matrix that does not explode with the size of the
    graph, then we can prove a bound for the entries of the fractional Laplacian using an argument similar to the undirected case (Proposition~\ref{prop: general decay}).
    This is completely analogous to the approach taken in~\cite{BB14} in the case of analytic functions of nonsymmetric matrices.
    Of course, now the constant $c$ in the bounds \eqref{eqn: decay for L^a} and \eqref{eqn: decay for P_a} should also include a bound for
    the condition number of an eigenvector matrix diagonalizing the Laplacian $L$.
    Another possibility is to give up the search for general bounds and to look at special cases
    for which we can find explicit (closed form) expressions for the
    entries of $L_{\text{out}}^\alpha$ (and their limit for $n\rightarrow\infty$), and from these
    obtain estimates for the probability of a given transition on the graph. This
    is the case of the \emph{directed cycle} and \emph{path graphs}; see Section
    \ref{sec:random_walks_on_directed_graphs}.

    We found that, for a large enough cycle, the transition
    probabilities  exhibit a power law decay parametrized by $\alpha$, in agreement with the bounds of Section \ref{sec:notes_on_the_decay_undirected}; see Section \ref{sec:random_walks_on_directed_graphs} for
    details. In particular, similar to what happens for undirected
    networks, we show in~Section \ref{sec:applications} that fractional diffusion-based random walks on directed graphs
    result in more efficient navigation of certain complex directed networks
    than using the local ones.

    \section{Superdiffusive processes on infinite graphs}\label{sec:superdiffusive_processes}

    In \cite{Estrada2017307,Estrada2018373} Estrada et al.~introduced a generalization of the diffusion equation on graphs, based on the $k$--path Laplacian, and they proved that the dynamics generated using the Mellin--transformed $k$--path Laplacian are superdiffusive processes on the infinite one-- and two--dimensional lattice graphs.
    In this section we exploit similar techniques to prove that the dynamics generated by the fractional Laplacians $L^\alpha$ are superdiffusive on an infinite one--dimensional graph, both in the undirected and directed case.

Consider a time--dependent probability distribution $u(t)_k$, $k \in
\Z$, such that $u(0)_k = \delta_{0k}$, i.e. the distribution at time
$t=0$ is concentrated in $0$. The mean square displacement (MSD) of
the distribution is defined as
    \begin{equation*}
    \text{MSD} = \langle {\lvert u(t) - u(0) \rvert}^2 \rangle = \sum_{k \in \mathbb{Z}} k^2 u(t)_k.
    \end{equation*}
    We say that a process is \textit{superdiffusive} if it generates probability distributions such that\footnote{We write $f(x) \sim g(x)$ for $x \to x_0$ if and only if $\lim_{x \to x_0} \frac{f(x)}{g(x)} = 1$ for both $x_0 \in \mathbb{R}$, and $x_0 = \pm\infty$.} $\text{MSD} \sim c t^\tau$ with $\tau > 1$ and $c>0$, for $t \to \infty$.
    In order to prove that the fractional diffusion dynamics on the infinite one--dimensional graph are superdiffusive, following the discussion in \cite{Estrada2017307}, we first show that by appropriately rescaling the solution $u(t)$, it converges to a stable probability distribution (Definition~\ref{def:stable_distribution}). Then, we will use some known properties of the limiting distribution to collect information on the behavior of the MSD of $u(t)$.

    \begin{definition}[Stable distribution]\label{def:stable_distribution}
        Let $\alpha \in {(0,2]}$, $\beta \in [-1,1]$, $\gamma >0$, $\delta \in \R$ and
        \begin{equation*}
        \omega(z, \alpha) = \begin{cases}
        -\tan (\alpha \frac{\pi}{2})    & \text{if } \alpha \ne 1, \\
        \frac{2}{\pi} \ln \abs{z} & \text{if } \alpha = 1.
        \end{cases}
        \end{equation*}
        A real random variable $X$ is called \textit{stable} if its characteristic function can be written as
        \begin{equation*}
        \mathbb{E}[e^{iz X}] = \phi(z; \alpha, \beta, \gamma, \delta) = \exp\left[i \delta z - \abs{\gamma z}^\alpha (1 + i \beta \sign(z) \omega(z, \alpha))\right].
        \end{equation*}
        This means that the density of $X$ is given by
        \begin{equation*}
        f(\xi; \alpha, \beta, \gamma, \delta) = \frac{1}{2\pi} \int_{-\infty}^{\infty} e^{-i \xi z} \, \phi(z; \alpha, \beta, \gamma, \delta) \dd z.
        \end{equation*}
    \end{definition}
    In the following, we will only use stable distributions with $\beta \in \{ 0, 1 \}$ and $\delta=0$, so we simplify the general notation to $f_\beta(\xi; \alpha, \gamma) \equiv f(\xi; \alpha, \beta, \gamma, \delta)$.

    \begin{subsection}{Undirected path graph}
        We start by examining the case of
        an \emph{infinite undirected path graph}, i.e. the graph $G = (V, E)$
        whose nodes are $V = \mathbb{Z}$ and whose edges are $E = \{ (k, k\pm 1) : k \in \mathbb{Z} \}$.
        In this case the adjacency and Laplacian matrices correspond respectively to the operators
        \begin{equation*}
        (A u)_k = u_{k-1} + u_{k+1}, \qquad u \in \ell^2(\mathbb{Z}),
        \end{equation*}
        and
        \begin{equation*}
        (Lu)_k = 2u_k - u_{k-1} - u_{k+1}, \qquad u \in \ell^2(\mathbb{Z}).
        \end{equation*}
        For $\alpha \in (0,1)$, we consider the fractional diffusion equation on $G$ with initial condition concentrated on the vertex indexed by $0${, i.e., the bi-infinite vector $e^{(0)}$ with 1 in position 0 and 0 everywhere else,}
        \begin{equation}
        \label{eq:fractional_diffusion_operator_form}
        \begin{cases}
        \displaystyle \ddt u(t) = -L^\alpha u(t), \\
        u(0) = e^{(0)}.
        \end{cases}
        \end{equation}
        As a first step, we find an explicit integral representation of the $k$th component $u(t)_k$ of
        the solution $u(t)$ of \eqref{eq:fractional_diffusion_operator_form}. This can be obtained by using the Fourier
        operator $\mathfrak{F} : \ell^2(\mathbb{Z}) \to L^2(-\pi, \pi)$ and its
        inverse $\mathfrak{F}^{-1} : L^2(-\pi, \pi) \to \ell^2(\mathbb{Z})$,
        \begin{align*}
        (\mathfrak{F} u) \, (\theta) = \frac{1}{\sqrt{2\pi}} \sum_{k \in \mathbb{Z}} e^{ik\theta} u_k, \quad u \in \ell^2(\mathbb{Z}), \qquad
        (\mathfrak{F}^{-1} g)_k = \frac{1}{\sqrt{2\pi}} \int_{-\pi}^{\pi} e^{-ikx} g(x) \dd x, \quad g \in L^2(-\pi, \pi).
        \end{align*}
        \begin{lemma}\label{lem:expression_of_the_solution}
            The solution $u(t)$ to~\eqref{eq:fractional_diffusion_operator_form} is given by
            \begin{equation}\label{eq:solution_of_operator_form_fractional_diffusion}
            u(t)_k = \frac{1}{2\pi} \int_{-\pi}^{\pi} e^{-ikx} \, e^{-t(2-2\cos x)^\alpha} \dd x, \qquad k \in \mathbb{Z}.
            \end{equation}
        \end{lemma}

        \begin{proof}
            It holds
            \begin{align*}
            (\mathfrak{F} Au) \,(\theta) &= \frac{1}{\sqrt{2\pi}} \sum_{k \in \mathbb{Z}} e^{ik\theta} (Au)_k
            = \frac{1}{\sqrt{2\pi}} \sum_{k \in \mathbb{Z}} e^{ik \theta} (u_{k-1} + u_{k+1})
            = \frac{1}{\sqrt{2\pi}} \sum_{k \in \mathbb{Z}} (e^{i(k-1)\theta} + e^{i(k+1)\theta}) u_k \\
            &=  (e^{-i\theta} + e^{i\theta}) \,( \mathfrak{F} u) (\theta)
            =  2 \cos\theta\,( \mathfrak{F} u) (\theta),
            \end{align*}
            and thus $(\mathfrak{F} L u) \, (\theta) = \left(2 - 2\cos\theta\right) \mathfrak{F} u (\theta)$. If we define $g = \mathfrak{F} u \in L^2(-\pi, \pi)$, then we have
            \begin{equation*}
            (\mathfrak{F} L \mathfrak{F}^{-1} g) \, (\theta) = (2-2\cos\theta) g(\theta).
            \end{equation*}
            We have therefore proved that $L$ is conjugated to the operator on $L^2(-\pi, \pi)$ that multiplies functions by $a(\theta) = 2-2\cos\theta$. In turn, this implies that $e^{-t L^\alpha}$ is conjugated to the multiplication by $a_\alpha(\theta) = e^{-t a(\theta)^\alpha}$.
            So, using the notation $g^{(0)} = \mathfrak{F} e^{(0)}$, the solution to \eqref{eq:fractional_diffusion_operator_form} can be expressed as
            \begin{align*}
            u(t)_k &= \left(e^{-tL^\alpha} u(0)\right)_k = \left(\mathfrak{F}^{-1} (a_\alpha g^{(0)})\right)_k
            = \frac{1}{\sqrt{2\pi}} \int_{-\pi}^{\pi} e^{-ikx} a_\alpha(x) g^{(0)}(x) \dd x \\
            &= \frac{1}{\sqrt{2\pi}} \int_{-\pi}^{\pi} e^{-ikx} e^{-t (2-2\cos x)^\alpha} ( \mathfrak{F} e^{(0)} ) (x)\dd x
            = \frac{1}{{2\pi}} \int_{-\pi}^{\pi} e^{-ikx} e^{-t(2-2\cos x)^\alpha} \sum_{n\in \mathbb{Z}} e^{inx} \, (e^{(0)})_n \dd x.
            \end{align*}
            Since the components of the initial condition are $(e^{(0)})_n = \delta_{0n}$, the previous expression simplifies to
            \begin{align*}
            u(t)_k &= \frac{1}{2\pi} \int_{-\pi}^{\pi} e^{-ikx} e^{-t(2-2\cos x)^\alpha} \dd x.\qedhere
            \end{align*}
        \end{proof}

        {We mention that Lemma~\ref{lem:expression_of_the_solution} can also be found as \cite[Theorem~1.3 (\textit{ii})]{LizamaRoncal}. We have included a proof in order to keep the discussion self-contained.}

        To prove that a properly scaled version of $u(t)$ converges to a stable distribution for $t \rightarrow \infty$, we need a Lemma from~\cite{Estrada2017307} linking together the expression of the solution~\eqref{eq:solution_of_operator_form_fractional_diffusion} and a stable distribution in Definition~\ref{def:stable_distribution} with $\beta=\delta=0$.
        However, we state it in a slightly more general formulation, which also introduces the asymmetry parameter $\beta \in [-1,1]$ and will be required shortly to deal with the directed case. Both the statement and the proof of the following lemma are based on Lemma~6.1 in \cite{Estrada2017307}.

        \begin{lemma}
            \label{lem:convergence_to_stable_distribution}
            Let $c>0$, $\alpha \in (0,2)$ and $\beta \in [-1,1]$ such that $\alpha \ne 1$ or $\beta = 0$.
            Let $h : [-\pi, \pi] \to \C$ be a continuous function that satisfies
            \begin{equation}
            \label{eqn: lemma stable hypothesis}
            \begin{alignedat}{2}
            &\text{Re}\,( h(x) ) > 0 && \text{for }x \in [-\pi, \pi] \setminus \{ 0 \},\\
            &h(x) \sim c \abs{x}^\alpha \left(\textstyle 1 - i \beta \sign(x) \tan(\alpha \frac{\pi}{2})\right) , \qquad &&\text{for } x \to 0.
            \end{alignedat}
            \end{equation}
            Then
            \begin{align}
            \label{eqn: lemma stable part-1}
            t^{1/\alpha} \frac{1}{2\pi} \int_{-\pi}^{\pi} e^{-i t^{1/\alpha} \xi x} \, e^{-t h(x)} \dd x \; &\to \; \frac{1}{2\pi} \int_{-\infty}^{\infty} e^{-i \xi z}\, e^{-c\abs{z}^\alpha \left(1-i \beta \sign(z) \tan(\alpha \frac{\pi}{2}) \right) } \dd z  \\
            \nonumber
            &= f_\beta(\xi; \alpha, c^{1/\alpha}),
            \end{align}
            uniformly in $\xi \in \R$ as $t \to \infty$. In other words,
            \begin{equation}
            \label{eqn: lemma stable part-2}
            \frac{1}{2\pi} \int_{-\pi}^{\pi} e^{-i t^{1/\alpha} \xi x} \, e^{-t h(x)} \dd x \; = \; t^{-1/\alpha} f_\beta(\xi; \alpha, c^{1/\alpha}) + o(t^{-1/\alpha}),
            \end{equation}
            uniformly in $\xi \in \R$ as $t \to \infty$.
        \end{lemma}

        \begin{proof}
            For any $\xi \in \R$ and $t>0$, using the substitution $z = t^{1/\alpha} x$, we have
            \begin{equation*}
            t^{1/\alpha} \frac{1}{2\pi} \int_{-\pi}^{\pi} e^{-i t^{1/\alpha} \xi x} e^{-t h(x)} \dd x = \frac{1}{2\pi} \int_{-\pi t^{1/\alpha}}^{\pi t^{1/\alpha}} e^{-i \xi z} e^{-t h(t^{-1/\alpha}z)} \dd z.
            \end{equation*}
            By substituting this in \eqref{eqn: lemma stable part-1} and using the triangle inequality, we get
            \begin{gather*}
            \left| \,  t^{1/\alpha} \frac{1}{2 \pi}\int_{-\pi}^{\pi} e^{-i t^{1/\alpha} \xi x} e^{-t h(x)} \dd x  - \frac{1}{2\pi} \int_{-\infty}^{\infty} e^{-i \xi z} e^{-c\abs{z}^\alpha \left( 1 - i\beta \sign(z)\tan(\alpha \frac{\pi}{2}) \right)} \dd z \,\right| \le \\
            \le \frac{1}{2\pi}  \int_{-\pi t^{1/\alpha}}^{\pi t^{1/\alpha}} \left| e^{-th(t^{-1/\alpha}z)} - e^{-c\abs{z}^\alpha \left( 1 - i\beta \sign(z)\tan(\alpha \frac{\pi}{2}) \right)} \right| \dd z + \frac{1}{2\pi} \int_{\R \setminus [-\pi t^{1/\alpha}, \pi t^{1/\alpha}]} e^{-c\abs{z}^\alpha} \dd z.
            \end{gather*}
            It is easy to see that the second term converges to $0$ as $t \to \infty$, so we only focus on the first term. Because of the hypothesis on the asymptotic behavior of $h(x)$, we have that
            \begin{equation*}
            \frac{t h(t^{-1/\alpha} z)}{ c \abs{z}^\alpha \left(1 - i \beta \sign(z) \tan(\alpha \frac{\pi}{2})\right)} \to 1 \qquad \text{as } t \to \infty.
            \end{equation*}
            This implies that, for any fixed $z \in \R$, the integrand in the first term goes to $0$ as $t \to \infty$. In order to conclude that the integral itself goes to $0$, by the Dominated Convergence Theorem it is sufficient to show that the integrand is bounded by an integrable function.

            Using the continuity of $h$ in conjunction with \eqref{eqn: lemma stable hypothesis}, it is not hard to see that there exists $\lambda > 0$ such that $\Real{h(x)} \ge \lambda \abs{x}^\alpha$. This implies that the integrand is bounded for all $t>0$ by the integrable function $f(z) = e^{-\lambda \abs{z}^\alpha} + e^{-c\abs{z}^\alpha}$, concluding the proof of \eqref{eqn: lemma stable part-1}. Note that the convergence is uniform in $\xi \in \R$, since the bounds we have obtained are independent of~$\xi$.
        \end{proof}

        {
        In order to have a cleaner statement for the next proposition, we allow the indices to be noninteger in the identity \eqref{eq:solution_of_operator_form_fractional_diffusion}; in other words, we write
        \begin{equation*}
            u(t)_z = \frac{1}{2\pi} \int_{-\pi}^{\pi} e^{-i z x} \, e^{-t(2-2\cos x)^\alpha} \dd x, \qquad \forall \, z \in \R.
        \end{equation*}
        }

        \begin{proposition}
            {By scaling} the solution $u(t)$ of~\eqref{eq:fractional_diffusion_operator_form} {with respect to $t$, it} converges to a stable probability distribution of the form $f_0(\xi; 2\alpha,1)$ for $t \to \infty$.
            {
            Specifically, for all $\xi \in \R$ it holds that
            \begin{equation*}
            t^{1/2\alpha}(u(t))_{t^{1/2\alpha} \xi} \to f_0(\xi; 2\alpha,1), \qquad \text{as } t \to \infty.
            \end{equation*}}
        \end{proposition}
        \begin{proof}
            From the expression of the solution in~\eqref{eq:solution_of_operator_form_fractional_diffusion},
            we can write it in the form
            \begin{equation*}
            u(t)_{{z}} = \int_{-\pi}^{\pi} e^{-i {z}x} e^{-th(x)} \dd x,
            \end{equation*}
            where
            $h(x) =
            a(x)^\alpha = (2-2\cos x)^\alpha = (x^2 + o(x^3))^\alpha =
            \abs{x}^{2\alpha} + o(x^{2\alpha+1})$ for $x \to 0$.
            Therefore, using Lemma~\ref{lem:convergence_to_stable_distribution} (with $\beta=0$) and the substitution $\xi = t^{-1/2\alpha} {z}$, we obtain
            \begin{equation*}
            t^{1/2\alpha} \left(u(t)\right)_{t^{1/2\alpha} \xi} \to f_0(\xi; 2\alpha, 1), \qquad t \to \infty,
            \end{equation*}
            or, equivalently,
            \begin{equation}
            \label{eqn: stable distribution limit}
            u(t)_{t^{1/2\alpha}\xi} = t^{-1/2\alpha} f_0(\xi;2\alpha,1) + o(t^{-1/2\alpha}) \qquad \text{for } t \to \infty.\qedhere
            \end{equation}
        \end{proof}

        We complete our analysis of the (behavior of the) solution for $t
        \to \infty$ by showing that the $\text{MSD} \sim c t^\tau$ with
        $\tau > 1$ and $c>0$, i.e., that we have superdiffusion. Observe now
        that, in our situation, the limiting stable distribution has an
        infinite variance since $2\alpha \in (0,2)$, thus we cannot compute
        the MSD of the solution directly. Let us look instead at the
        asymptotic behavior of the square of the full width at half maximum
        (FWHM) of the solution, since $\text{FWHM}^2$ gives a lower bound
        for the MSD;
        we recall that the FWHM can be defined as
        \begin{equation*}
        \text{FWHM} = \max \left\lbrace \abs{b-a} : f(\xi) \ge {\textstyle\frac{1}{2}} \max_{x\in \R} f(x), \; \forall\, \xi \in [a,b] \right\rbrace.
        \end{equation*}
        \begin{theorem}
            \label{thm:undirected_superdiffusion}
            The fractional diffusion process on the infinite undirected path graph is superdiffusive for all $\alpha \in (0,1)$. In particular, the mean square displacement of the solution satisfies $\text{MSD} \ge \tilde c t^{1/\alpha}$, as $t \to \infty$.
        \end{theorem}

        \begin{proof}
            Let $\xi_0 \in \R$ be such that $f_0(\xi_0; 2\alpha, 1) = \frac{1}{2} f_0(0; 2\alpha,1)$, so that the full width at half maximum of the distribution $f$ is $\text{FWHM}\{f(\xi)\} = 2 \xi_0$. Recalling equation \eqref{eqn: stable distribution limit} and using the fact that the FWHM is invariant under vertical scalings, we have that
            \begin{align*}
            \text{FWHM}\{ u(t)_k \} &= t^{1/2\alpha} \, \text{FWHM} \{u(t)_{t^{1/2\alpha} k}\} \\
            &= t^{1/2\alpha} \, \text{FWHM} \{ t^{-1/2\alpha} f_0(k; 2\alpha,1) + o(t^{-1/2\alpha}) \} \\
            &= t^{1/2\alpha} \, \text{FWHM} \{ f_0(k;2\alpha,1) + o(1) \} \\
            &\sim t^{1/2\alpha} \, 2\xi_0, \quad \qquad t \to \infty.
            \end{align*}
            Therefore $\text{FWHM}^2 \sim 2 \xi_0 \, t^{1/\alpha}$ and, since $\alpha \in (0,1)$,
            we have that $\text{FWHM}^2 \sim c \, t^\tau$ with $\tau>1$. Thus we also have $\text{MSD} \ge \tilde c \, t^\tau$,
            i.e., the process is superdiffusive.
        \end{proof}

    \end{subsection}

    \begin{subsection}{Directed path graph}

In this part, we perform the same analysis for the fractional
diffusion equation on the \textit{infinite directed path graph},
i.e. the graph $G=(V,E)$ with nodes $V = \Z$ and edges $E = \{
(k,k+1) : k \in \Z \}$. Similarly to the undirected case, the
solution converges to a stable distribution when appropriately
scaled, and we can use this fact to describe the behavior of the MSD
of the solution for $t \to \infty$.

We first observe that on a directed graph the diffusion equation
uses the transpose of the nonsymmetric Laplacian $L_{\text{out}}$
instead of $L_{\text{out}}$.
{Indeed, the solution for the dynamics induced by $L_\text{out}^T$ remains a probability vector at all times since
$L_\text{out}^T \boldsymbol{1} = \boldsymbol{0}$; on the other hand, this property is not preserved by the dynamics induced by $L_\text{out}$, since in general
$\boldsymbol{1}^T L_\text{out} \ne \boldsymbol{0}^T$.} Using $L$ to denote the transpose of
the out-degree Laplacian of $G$ for simplicity of notation, the
fractional diffusion equation on a directed graph is
        \begin{equation}
        \label{eqn: directed fractional diffusion}
        \begin{cases}
        \displaystyle \ddt u(t) = - L^\alpha u(t), \qquad \alpha \in (0,1), \\
        u(0) = e^{(0)},
        \end{cases}
        \end{equation}
        where the initial condition is the one with all the mass concentrated on the vertex $0$, i.e. $(e^{(0)})_k = \delta_{0k}$. The (transposes of the) adjacency and Laplacian matrices correspond respectively to:
        \begin{equation*}
        \begin{aligned}
        (A u)_k &= u_{k-1}, &\quad u \in \ell^2(\Z), \\
        (L u)_k &= u_k - u_{k-1}, &\quad u \in \ell^2(\Z).
        \end{aligned}
        \end{equation*}
        \begin{lemma}\label{lem:solution_for_directed_graph_path}
            The solution $u(t)$ to~\eqref{eqn: directed fractional diffusion} is given by
            \begin{equation}
            \label{eqn: directed solution}
            u(t)_k = \frac{1}{2\pi} \int_{-\pi}^{\pi} e^{-ikx} e^{-t(1-e^{ix})^\alpha} \dd x.
            \end{equation}
        \end{lemma}
        \begin{proof}
            It holds
            \begin{align*}
            (\mathfrak{F} Au) \,(\theta) &= \frac{1}{\sqrt{2\pi}} \sum_{k \in \Z} e^{ik\theta} (Au)_k = \frac{1}{\sqrt{2\pi}} \sum_{k \in \Z} e^{ik \theta} u_{k-1} \\
            &= \frac{1}{\sqrt{2\pi}} \sum_{k \in \Z} e^{i(k+1)\theta} u_k =  e^{i\theta} \,( \mathfrak{F} u) (\theta).
            \end{align*}
            Therefore we get $(\mathfrak{F} L u) \, (\theta) = \left(1 - e^{i\theta} \right) \mathfrak{F} u (\theta)$. If we define $g = \mathfrak{F} u \in L^2(-\pi, \pi)$, we have
            \begin{equation*}
            (\mathfrak{F} L \mathfrak{F}^{-1} g) \, (\theta) = (1 - e^{i\theta}) g(\theta).
            \end{equation*}
            So $L$ is conjugated to the operator on $L^2(-\pi, \pi)$ that multiplies functions by $a(\theta) = 1 - e^{i\theta}$, and this implies that $e^{-t L^\alpha}$ is conjugated to the multiplication by $a_\alpha(\theta) = e^{-t a(\theta)^\alpha}$. Using the notation $g^{(0)} = \mathfrak{F} e^{(0)}$, we can write the solution to \eqref{eqn: directed fractional diffusion} explicitly in the form
            \begin{align*}
            \nonumber
            u(t)_k &= \left(e^{-tL^\alpha} e^{(0)}\right)_k = \left(\mathfrak{F}^{-1} (a_\alpha g^{(0)})\right)_k = \frac{1}{\sqrt{2\pi}} \int_{-\pi}^{\pi} e^{-ikx} a_\alpha(x) g^{(0)}(x) \dd x \\
            &= \frac{1}{\sqrt{2\pi}} \int_{-\pi}^{\pi} e^{-ikx} e^{-t (1-e^{ix})^\alpha} \mathfrak{F} e^{(0)} (x)\dd x = \frac{1}{{2\pi}} \int_{-\pi}^{\pi} e^{-ikx} e^{-t(1-e^{ix})^\alpha} \sum_{n\in \Z} e^{inx}(e^{(0)})_n \dd x \\
            &= \frac{1}{2\pi} \int_{-\pi}^{\pi} e^{-ikx} e^{-t(1-e^{ix})^\alpha} \dd x.
            \end{align*}
        \end{proof}

        {The result in Lemma~\ref{lem:solution_for_directed_graph_path} is a particular instance of a question with a long history concerning the ``non-integer orders of summability'' in the Ces\`{a}ro sense; see, e.g., the seminal paper by Chapman~\cite[Parts~III, and IV]{Chapman}. The question of the convergence of $-L^\alpha u$ for general sequences of complex numbers have been addressed in~\cite[Theorem~1]{Kuttner}. Thus, although the expression in~\eqref{eqn: directed solution} was already known, see the discussion in~\cite[Section~1]{Abadias}, we decided to give it here explicitly and with full details for the sake of keeping the discussion self-contained.}

        Note that as $x \to 0$ we have $(1 - e^{ix})^\alpha \sim (-i x)^\alpha \sim \abs{x}^\alpha \left(\cos(\alpha \frac{\pi}{2}) - i \sign(x) \sin(\alpha \frac{\pi}{2})\right)$.

        We can now use Lemma~\ref{lem:convergence_to_stable_distribution} to prove that the solution $u(t)$ converges to a stable distribution if appropriately scaled.
        {
        Similar to what we did in the undirected case, for ease of notation we expand identity \eqref{eqn: directed solution} to also include noninteger indices; that is, we write
        \begin{equation*}
            u(t)_z = \frac{1}{2\pi} \int_{-\pi}^{\pi} e^{-izx} e^{-t(1-e^{ix})^\alpha} \dd x, \qquad \forall \, z \in \R.
        \end{equation*}
        }

        \begin{proposition}
            {By scaling the} solution $u(t)$ of \eqref{eqn: directed fractional diffusion} {with respect to $t$, it} converges to a stable probability distribution of the form $f_1(\xi; \alpha, c)$ for $t \to \infty$, {where $c = \cos (\alpha \frac{\pi}{2})$}. {Specifically, for all $\xi \in \R$ it holds that
            \begin{equation*}
            t^{1/\alpha}(u(t))_{t^{1/\alpha}\xi} \to f_1(\xi; \alpha, c), \qquad \text{as } t \to \infty.
            \end{equation*}}
        \end{proposition}
        \begin{proof}
            In the expression for the solution \eqref{eqn: directed solution} we have
            \begin{equation*}
            h(x) = (1 - e^{ix})^\alpha \sim \textstyle\cos({\alpha} \frac{\pi}{2})\abs{x}^\alpha \left( 1 - i \sign(x) \tan(\alpha \frac{\pi}{2}) \right), \qquad x \to 0.
            \end{equation*}
            Using Lemma \ref{lem:convergence_to_stable_distribution} and introducing for simplicity of notation $c = \cos({\alpha} \frac{\pi}{2})^{1/\alpha}$, we get that for all $\xi \in \R$,
            \begin{equation*}
            t^{1/\alpha} \left(u(t)\right)_{t^{1/\alpha} \xi} \to f_1(\xi; \alpha, c), \qquad t \to \infty,
            \end{equation*}
            or equivalently
            \begin{equation}
            \label{eqn: stable distribution limit 2}
            u(t)_{t^{1/\alpha}\xi} = t^{-1/\alpha} f_1(\xi; \alpha, c) + o(t^{-1/\alpha}), \qquad t \to \infty.
            \end{equation}
        \end{proof}

        As in the undirected case, the limiting stable distribution has an infinite variance since $\alpha < 2$, so we cannot compute the MSD of the solution directly, and we instead examine the behavior of the square of the FWHM of the solution.
        \begin{theorem}
            \label{thm:directed_superdiffusion}
            The full width at half maximum of the solution of the fractional diffusion process on the infinite directed path graph satisfies $\text{FWHM}^2 \sim \tilde c t^{2/\alpha}$, as $t \to \infty$.
        \end{theorem}
        \begin{proof}
            Let $\xi_0 \in \R$ be such that $f_1(\xi_0; \alpha, c) = \frac{1}{2} f_1(0; \alpha, c)$, so that the full width at half maximum of the distribution $f_1$ is $\text{FWHM}\{f_1(\xi)\} = \xi_0$ (note that the density $f_1$ is nonsymmetric and identically $0$ for $\xi < 0$).
            Recalling equation \eqref{eqn: stable distribution limit 2} and using the fact that the FWHM is invariant under vertical scalings, we have
            \begin{align*}
            \text{FWHM}\{ u(t)_k \} &= t^{1/\alpha} \, \text{FWHM} \{u(t)_{t^{1/\alpha} k}\} \\
            &= t^{1/\alpha} \, \text{FWHM} \{ t^{-1/\alpha} f_1(k; \alpha,c) + o(t^{-1/\alpha}) \} \\
            &= t^{1/\alpha} \, \text{FWHM} \{ f_1(k;\alpha, c) + o(1) \} \\
            &\sim t^{1/\alpha} \, \xi_0, \quad \qquad t \to \infty.
            \end{align*}
            Therefore we obtain $\text{FWHM}^2 \sim  \xi_0^2 \, t^{2/\alpha}$.
        \end{proof}

        Note that, in contrast to Theorem~\ref{thm:undirected_superdiffusion}, with Theorem~\ref{thm:directed_superdiffusion} we have proved that the fractional diffusion dynamics on the infinite directed path graph is ``superdiffusive'' for all $\alpha \in (0,2)$; in particular, this holds also for classical diffusion, $\alpha = 1$. This behavior seems at first sight confusing, but it can be explained by observing that the
        interpretation of~\eqref{eqn: directed fractional diffusion} as describing a diffusion process is not appropriate. Indeed, the probability distribution is not really subjected to a diffusion process, since it is always ``pushed'' in the same direction in the graph; in other words, this process is more similar to a fractionalization of advection (or transport) than of diffusion. This can also be observed by comparing the definitions of the Laplacians of the undirected and directed path graphs: while the former one corresponds to a centered discretization of the second derivative in space (\textit{diffusion}), the latter one corresponds to a forward discretization of the first derivative in space (\textit{advection}).

        In conclusion, we have proved that the solution to the fractional ``diffusion'' dynamics~\eqref{eqn: directed fractional diffusion} on the directed path graph expands faster than the classical dynamics, similarly to what we proved in the undirected case; however, we cannot directly compare the directed case with the undirected one, since they can be respectively interpreted as advection and diffusion, and thus they have different time scales.
    \end{subsection}

    \section{Closed form expressions for two simple cases}
    \label{sec:random_walks_on_directed_graphs}

    Having defined the fractional $\alpha$th power of the matrix $L_{\text{out}}$, we consider the normalized
    version of $\bar{L}_{\text{out}}^{(\alpha)}$ with entries
    $(L_{\text{out}}^\alpha)_{i,j}/(L_{\text{out}}^\alpha)_{i,i}$. It
    can then be exploited to generate the discrete time dynamics of a random walker on
    a directed graph by considering the transition matrix
    $P_{\text{out}}^{(\alpha)} = I - \bar{L}_{\text{out}}^{(\alpha)}$.
    As in the symmetric case discussed
    in~\cite{PhysRevE90032809} and in Lemma~\ref{lemma:fractional random walk}, this matrix is a row stochastic
    matrix, and the standard transition matrix for the Laplacian
    $L_{\text{out}}$ is recovered as $\alpha \rightarrow 1$.

    To completely describe the
    behavior of the random walker in a fully analytical setting we consider two test cases,
    the directed path $\mathcal{P}_n$, and the directed cycle graph $\mathcal{C}_n$.

    The \textit{directed path} $\mathcal{P}_n$ is the graph with adjacency matrix $A = (a_{i,j})$
    with $a_{i,i+1} = 1$, $i=1,\ldots,n-1$, and whose outdegree Laplacian $L_{\text{out}}$
    is
    \begin{equation*}
    L_{\text{out}} = \begin{bmatrix}
    1 & -1 \\
    &  1 & -1 \\
    &    &  \ddots & \ddots \\
    &    &         & 1 & -1 \\
    &    &         &   &    0
    \end{bmatrix}.
    \end{equation*}
    This is a nonsymmetric, nondiagonalizable matrix, thus we cannot apply
    decomposition~\eqref{eq:fractional_laplacian_of_a_graph}, and we need
    to use Definition~\ref{def:matrix_function_as_jordan_form}. Therefore, we first
    need to compute the Jordan canonical form of $L_{\text{out}} = Z J Z^{-1}$, that reads as
    \begin{equation*}
    Z = \begin{bmatrix}
    1 & -1 &  &  \\
    1 &  & 1 &  \\
    \vdots &  &  & \ddots \\
    1 &  &  &  &  & (-1)^{n-1} \\
    1 &  &  &  &  &  0
    \end{bmatrix}, \quad  J = \begin{bmatrix}
    0 &  &    \\
    & 1 & 1  \\
    & & \ddots & \ddots \\
    & &  & \ddots & 1  \\
    & &  & &   1 \\
    \end{bmatrix}.
    \end{equation*}
    Thus, the resulting matrix function can be expressed by computing
    \begin{equation*}
    J^{\alpha} = \begin{bmatrix}
    0 & 0 & \cdots  & 0  \\
    & \binom{\alpha}{0} & \cdots  & \binom{\alpha}{n-1}  \\
    & & \ddots & \vdots \\
    &  & & \binom{\alpha}{0} \\
    \end{bmatrix}, \quad \binom{\alpha}{k} = \frac{\alpha\cdot\ldots\cdot (\alpha-k+1)}{k!},
    \end{equation*}
    and by expressing $L_{\text{out}}^\alpha = Z J^\alpha Z^{-1}$. So for $h \geq 1$ and $k < n$ we can express its $(h,k)$ element as
    \begin{equation*}
    (L_{\text{out}}^\alpha)_{h,k} = \begin{cases}
    0,                              & \text{ if } k < h \text{ or } k = h = n, \\
    -1,                             & \text{ if }(h,k) = (n-1,n), \\
    (-1)^{h+k} \binom{\alpha}{k-h}, & \text{ if } 1 \le h \le k \le n-1.
    \end{cases}
    \end{equation*}
    Therefore, the probability $p_{h\to k}^{(\alpha)}$ of a transition $h \to k$
    on the directed path graph is given by
    \begin{equation*}
    p_{h \to k}^{(\alpha)} = \begin{cases}
    0, & h = n, \\
    \delta_{h,k} - \frac{(L_{\text{out}}^\alpha)_{h,k}}{(L_{\text{out}}^\alpha)_{h,h}} = \delta_{h,k} - (L_{\text{out}}^\alpha)_{h,k}, &\text{otherwise}.
    \end{cases}
    \end{equation*}
    If we let the size of the graph $n$ grow to infinity, and consider the decay of the transition probability for large values of $k > h$ we observe that
    {\begin{equation*}
    p_{h \to k}^{(\alpha)} =  - \frac{\Gamma(k-h-\alpha)}{\Gamma(-\alpha)\Gamma(k-h+1)}
\sim  \frac{\Gamma (\alpha +1)\sin(\pi \alpha)  }{\pi }k^{-\alpha -1}, \quad \text{ since } \frac{\Gamma(x+\alpha)}{\Gamma(x+\beta)} \sim x^{a-b} \text{ as } x \rightarrow +\infty,
    \end{equation*}}
    i.e., a polynomial decay parameterized by $\alpha$. Note that the associated chain has an absorbing state (the last vertex), which is always reached. Therefore, the effect of the nonlocality is reflected by the fact that we have a higher probability of transitioning to a far away node without completely exploring the network. In Figure~\ref{fig:path_markov_simulation}, we observe the simulated behavior for 10 steps on a directed path with $n=20$ nodes, always starting from the first one.
    \begin{figure}[htbp]
        \centering
        \includegraphics[width=0.8\columnwidth]{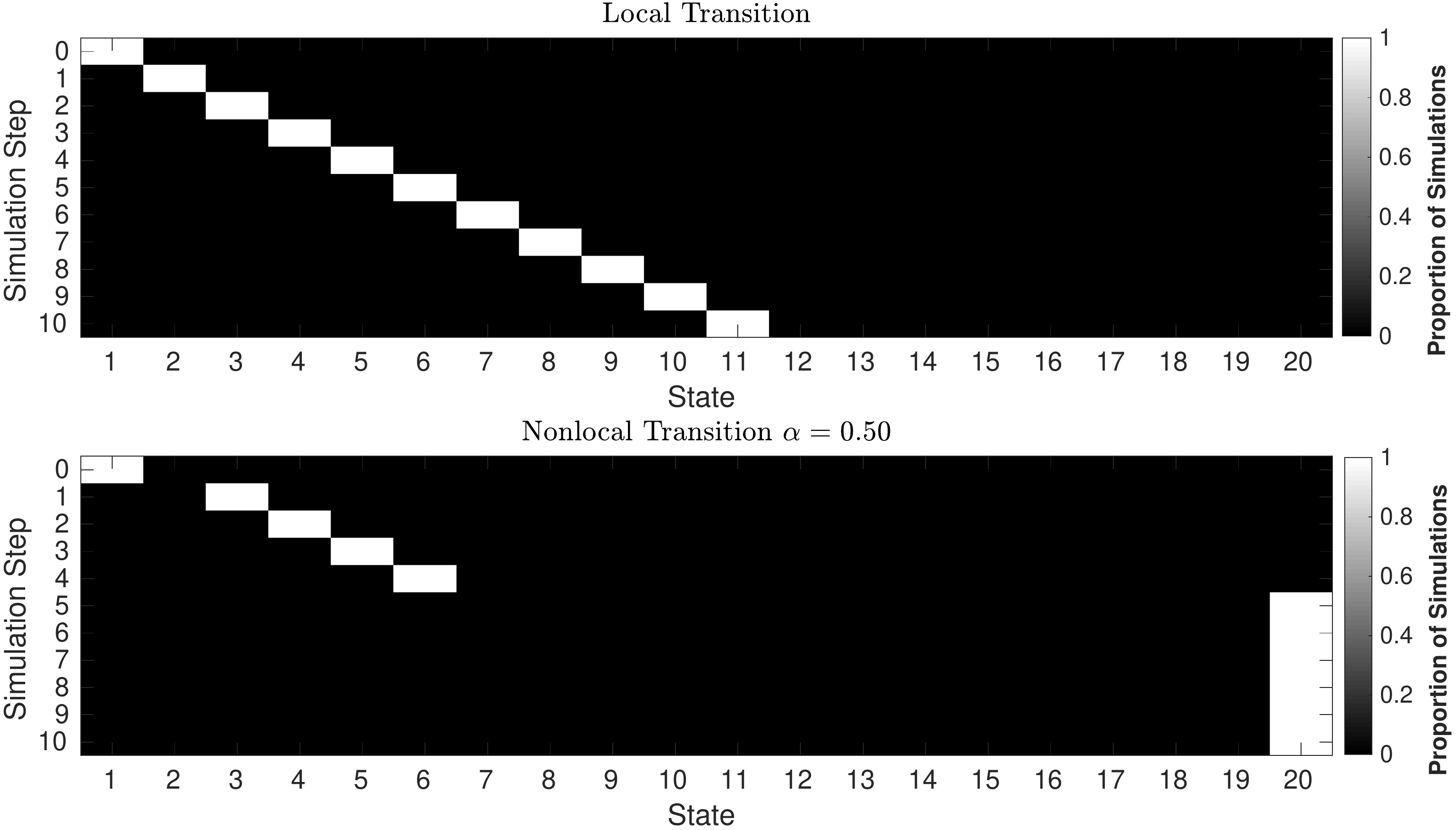}
        \caption{Simulation of $10$ steps of the local and nonlocal ($\alpha=0.5$) Markov chains on the directed path with $n=20$ nodes. In both cases we start from the first node of the chain, when $\alpha=0.5$ we reach the absorbing state in $5$ steps, which is exactly the value predicted by~\eqref{eq:number_of_step_to_absorption}.}
        \label{fig:path_markov_simulation}
    \end{figure}
    Moreover, decreasing the value of $\alpha$ resolves in faster absorption. To compute the average number of steps needed to reach the absorbing state starting from the first node, we partition the matrix $L^{\alpha}_\text{out}$ into the block form
    \begin{equation*}
    L^{\alpha} = \begin{bmatrix}
    I - Q & \mathbf{r} \\
    \mathbf{0} & 0
    \end{bmatrix},
    \end{equation*}
    to extract the inverse of the fundamental matrix $I-Q$. Then we can compute the expected number of steps $n_{\text{step}}$~\cite[Theorem 3.3.5]{MR0115196} as $n_{\text{step}} = \lceil(I-Q)^{-1} \mathbf{1})_{1}\rceil$.  By reusing the computation done for $J^{\alpha}$, it is easy to prove that $(I-Q)^{-1}$ is the upper triangular Toeplitz matrix with first row $(\mathbf{t})_\ell = t_\ell = (-1)^{\ell-1} \binom{-\alpha}{\ell-1}$, $\ell=1,\ldots,n-1$. Therefore, the expected number of steps needed to reach the absorbing state starting from the first node is
    \begin{equation}\label{eq:number_of_step_to_absorption}
    \begin{split}
    n_{\text{step}} = &\; \left\lceil\sum_{\ell=1}^{n-1}(-1)^{\ell-1} \binom{-\alpha}{\ell-1}\right\rceil =  \left\lceil\frac{(-1)^{n+1} (n-1) \binom{-\alpha }{n-1}}{\alpha }\right\rceil,
    \end{split}
    \end{equation}
    which is a monotonically increasing function with respect to $\alpha$.

    For the case of the \textit{directed cycle} graph $\mathcal{C}_n$, i.e., of the graph with nodes
    $V = \{1,\ldots,n\}$ and directed edges  $E = \{(j,j+1), \,
    j=1,\ldots,n-1 \} \cup \{(n,1)\}$, the out-degree Laplacian is then
    the circulant matrix of size $n$ with first row
    $\mathbf{c}=[1,-1,0,\ldots,0]$, i.e., $(L_{\text{out}})_{i,j} =
    c_{j-i \text{(mod }n\text{)}}$. This is a normal matrix which
    is diagonalized by the discrete Fourier matrix of size $n$, $F_n$,
    and whose eigenvalues are given by $\lambda_\ell(L_{\text{out}}) = 1 -
    \exp(-{2 \ell i \pi}/{n})$. By
    using~\eqref{eq:fractional_laplacian_of_a_graph} for this particular
    out-degree Laplacian we find
    \begin{equation*}
    (L_{\text{out}}^\alpha)_{h,k} = \frac{1}{n} \sum_{\ell = 1}^{n}\left( 1 - e^{\frac{-2 \ell \pi \cdot i}{n}} \right)^\alpha e^{\frac{2 \ell \pi\cdot i}{n} (h-k)}.
    \end{equation*}
    Taking the limit for $n \rightarrow + \infty$ we can then express
    the~$(h,k)$ element of $L_{\text{out}}^\alpha$ as
    \begin{equation*}
    \begin{split}
    (L_{\text{out}}^\alpha)_{h,k} = \;\frac{1}{2\pi} \int_{0}^{2\pi} (1-e^{-i \theta})^{\alpha }e^{ i d_{h,k} \theta}\,d\theta  = \; \frac{\Gamma (d_{h,k}-\alpha )}{d_{h,k}! \Gamma (-\alpha )},
    \end{split}
    \end{equation*}
    where $d_{h,k} = h-k \,(\operatorname{mod}\,n)$. Therefore, the probability $p_{h\to k}^{(\alpha)}$ of a $h\to k$ transition
    on the cycle graph is given by
    \begin{equation*}
    p_{h \to k}^{(\alpha)} = \delta_{h,k} - \frac{(L_{\text{out}}^\alpha)_{h,k}}{(L_{\text{out}}^\alpha)_{h,h}} = \delta_{h,k}-\frac{\Gamma (d_{h,k}-\alpha )}{d_{h,k}! \Gamma (-\alpha )}.
    \end{equation*}
    For $h,k$ such that $d_{h,k} \gg 1$ we can expand this transition probability, for $\alpha \in (0,1)$, as
    \begin{equation*}
    p_{h \to k}^{(\alpha)} = -d_{h,k}^{-\alpha }\left(\frac{1}{d_{h,k} \Gamma (-\alpha )}+O\left(\frac{1}{d_{h,k}^2}\right)\right) \approx -\frac{d_{h,k}^{-\alpha -1}}{\Gamma (-\alpha )},
    \end{equation*}
    thus showing that, for a large enough cycle, the transition
    probability behaves as a distribution whose probabilities decay
    polynomially with respect to $\alpha$. In this case the underlying
    graph is strongly connected, therefore we do not have any absorbing
    states in the chain. In the local dynamics case, we can
    be sure that in a number of steps equal to the number of nodes of
    the network we completely explore it, while, on the other hand, the
    possibility of performing longer jumps increases the probability of
    returning to certain states while leaving others untouched. See,
    e.g., the example in Figure~\ref{fig:cycle_markov_simulation} for a
    directed cycle graph with $n=20$ nodes in which $10$ jumps are
    performed.
    \begin{figure}[htbp]
        \centering
        \includegraphics[width=0.8\columnwidth]{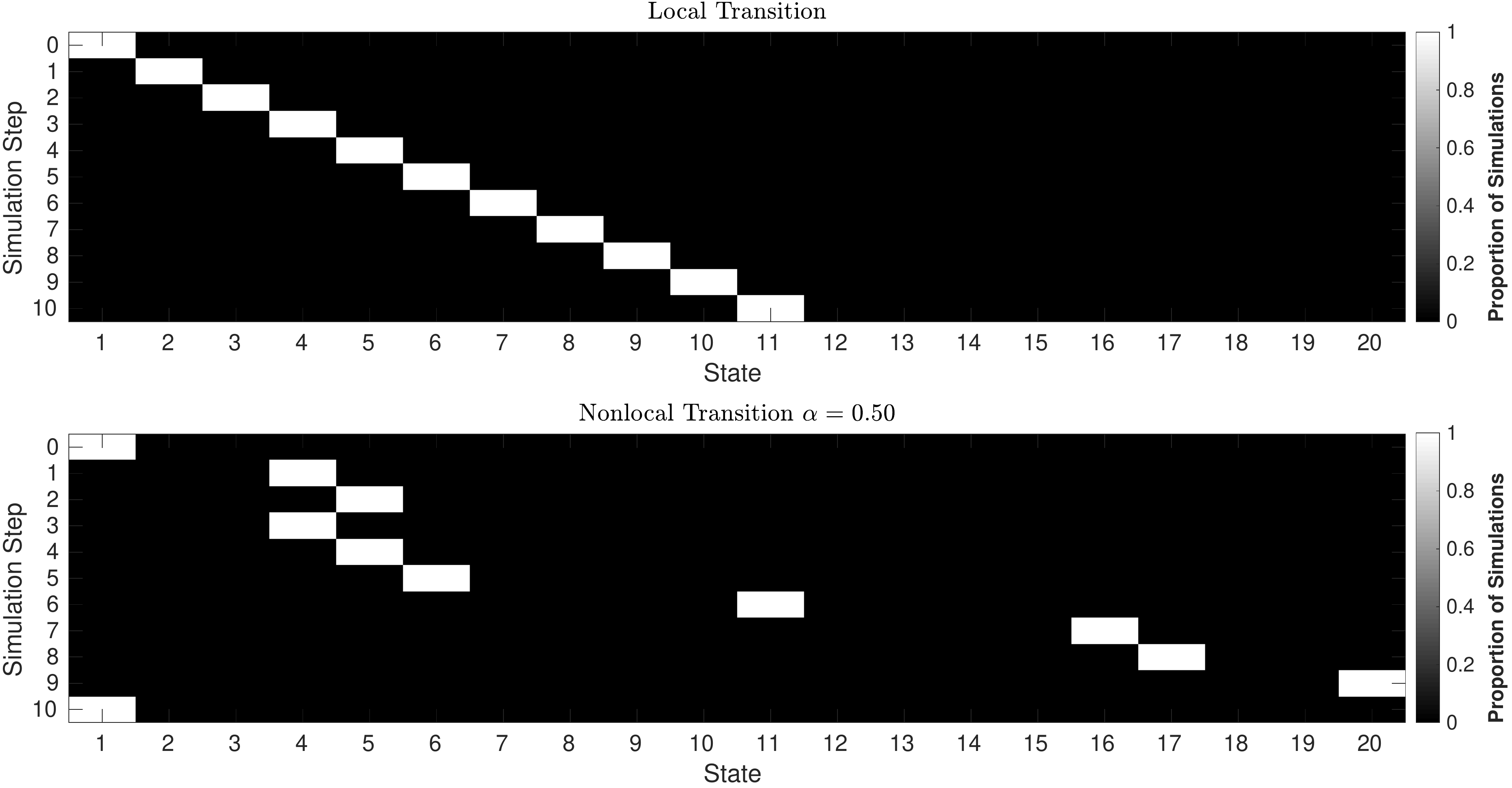}
        \caption{Simulation of $10$ steps of the local and nonlocal
            ($\alpha=0.25$) Markov chains on the directed cycle graph with
            $n=20$ nodes. In both cases we start from the same node of the cycle
            we identifiy with node~1.} \label{fig:cycle_markov_simulation}
    \end{figure}

    \section{Applications}\label{sec:applications}

    The simple examples from
    Section~\ref{sec:random_walks_on_directed_graphs} seem to suggest
    that, in the presence of a strong directionality in the network, the
    possibility of performing long distance jumps does not necessarily
    lead to better (i.e., faster) exploration of the network compared to
    the classical, local random walk (or, in the case of continuous
    time, diffusion) dynamics.  Real world directed networks, however,
    are very different from these simple ``unidirectional'' graphs,
    and allow for far richer exploration dynamics. To understand what we
    have gained in moving from the standard random walk on the network
    to its fractional extension, we consider the efficiency of the new
    dynamics in exploring the underlying directed graph compared to the
    classical dynamics. To measure it, we consider the average return
    probability at time $t$, $p_0^{(\alpha)}(t)$, for a continuous time
    random walker described by the master equation for the probability
    $p(i,t|i_0,0)$ of being at node $i$ at time $t$ having started from
    node $i_0$ at time $t = 0$, for the dynamics induced by the
    normalized version of $\bar{L}_{\text{out}}^{(\alpha)}$. The
    continuous time random walk master equation on a directed graph
    reads as
    \begin{equation*}
    \partial_t p(i,t|i_0,0) = - \sum_j p(j,t|i_0,0) (\bar{L}_{\text{out}}^{(\alpha)})_{ji} ,
    \end{equation*}
    with initial condition $p(i,0|i_0,0) = \delta_{i,i_0}$. The desired
    average return probability is obtained as
    $$
    p_0^{(\alpha)}(t) = \frac{1}{n} \sum_{i=1}^n p(i,t|i,0) =
    \frac{1}{n} \sum_{i=1}^{n}
    \exp(-\lambda_i(\bar{L}_{\text{out}}^{(\alpha)})t).
    $$
    Even if $\bar{L}_{\text{out}}^{(\alpha)}$ has complex eigenvalues,
    they always appear in conjugate pairs. Therefore,
    $p_0^{(\alpha)}(t)$ is always a real number; specifically, we
    consider the network \texttt{Roget} in which each vertex corresponds
    to one of the categories in the 1879 edition of Peter Mark ``Roget's
    Thesaurus of English Words and Phrases'', and in which each arc
    connects two categories whenever Roget give reference to one of them
    among the words and phrases of the other, or if the two categories
    are directly related by their positions in the book. The
    \texttt{wiki-Vote} network containing all the Wikipedia voting data
    from the 2,794 elections that had taken place till January 2008.
    Nodes in the network represent Wikipedia users, directed arcs from
    node $i$ to node $j$ exists whenever user $i$ voted for user $j$.
    The network \texttt{p2p-Gnutella08} obtained from the eight of the
    nine snapshots of the Gnutella peer-to-peer file sharing network
    collected in August 2002. In this case, the nodes are the hosts in
    the Gnutella network topology and the arcs are the connections
    between the hosts. For all the three cases, we restrict to the
    largest connected component of the network. The following examples
    demonstrate that in the case of real world complex digraphs, the
    use of nonlocal diffusion processes (or random walks) display
    similar advantages to those observed in the undirected case. We
    report in Figure~\ref{fig:return_probability} the quantity
    $p_0^{(\alpha)}(t)$ while highlighting the value of the first
    nonzero eigenvalue of the associated Laplacian. For each network
    we also report the \textit{relative} spectral gap (magnitude of the ratio of the largest
    to the smallest nonzero eigenvalue) and the network diameter.

    As we can
    observe, the higher the spectral gap, i.e., the larger the modulus
    of the second smallest eigenvalue of the Laplacian matrix
    $L_{\text{out}}$ is, the more efficient the fractional exploration
    of the associated network is. This is an expected behavior since the
    average return probability is directly linked to the whole spectral
    distribution of the associated normalized Laplacian matrix. In
    particular, it is well known that sparse networks with larger
    spectral gap can be explored more efficiently than those having a
    smaller spectral gap. The network's diameter, on the other hand, seems to be
    less relevant as an indicator of when the nonlocal dynamics is more efficient than
    the local one.

    \begin{figure}[htb]
        \centering
        \includegraphics[width=0.80\columnwidth]{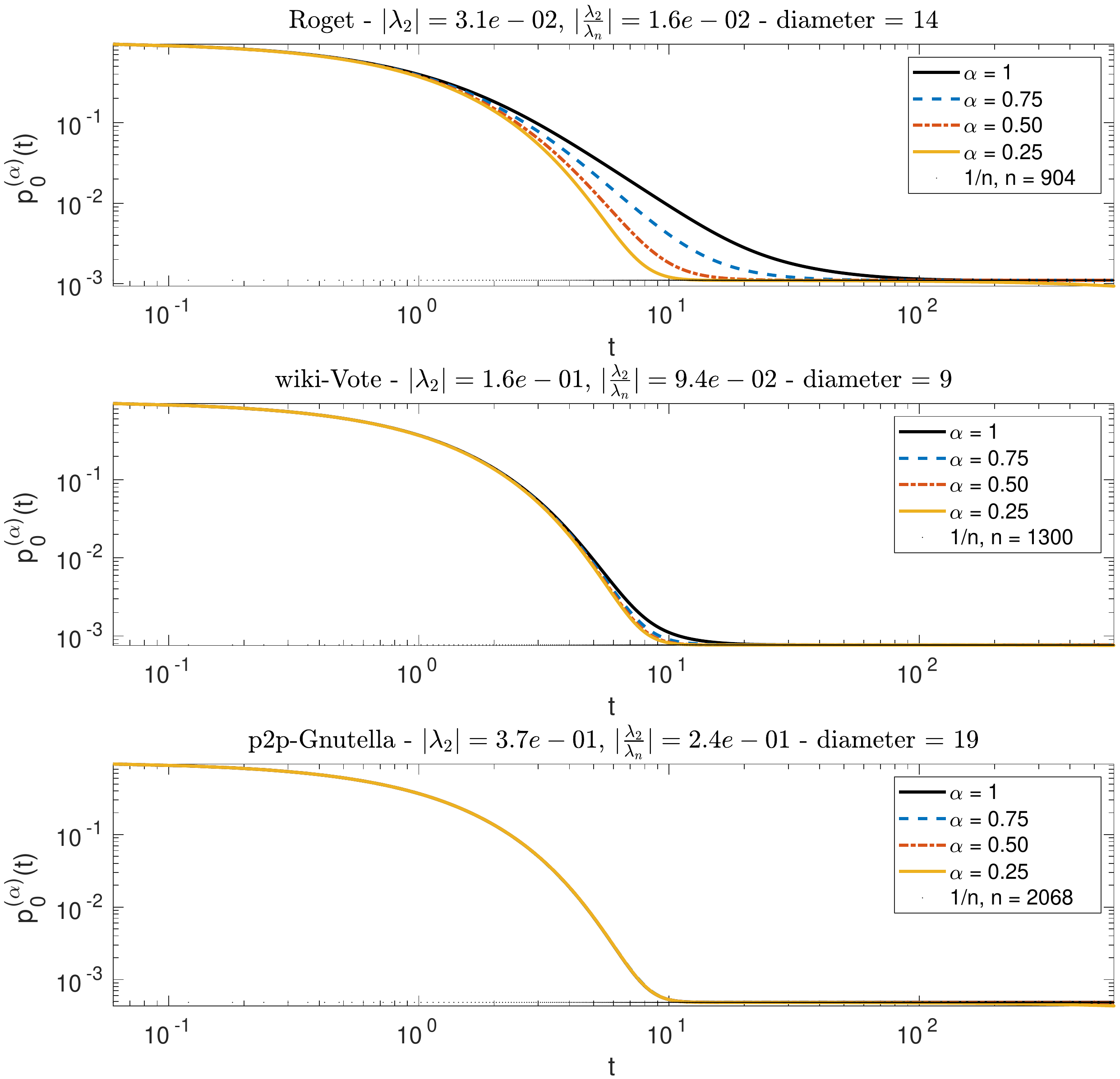}
        \caption{(Color online) Average fractional return probability
            $\mathbf{p}_{i}^{(\alpha)}(t)$ as a function of time for three
            different directed networks \texttt{Roget}, \texttt{wiki-Vote}, and
            \texttt{p2p-Gnutella08}. In each case we restrict the analysis to
            the largest connected component of the
            network. Note the logarithmic scale on the axes.}\label{fig:return_probability}
    \end{figure}

    We also observe that the behavior shown in
    Figure~\ref{fig:return_probability} is similar to that observed for
    the fractional dynamics on undirected networks
    in~\cite{PhysRevE90032809}.

    \subsection{Consensus models for control of vehicle motions}\label{sec:consensus_models}

    Consider an ensemble of $N$ vehicles moving in an $m$th dimensional
    space. We denote the initial positions by $\mathbf{x}_i \in
    \mathbb{R}^m$, $i=1,\ldots,N$, and the initial velocities by
    $\mathbf{v}_i \in \mathbb{R}^m$, $i=1,\ldots,N$. We are interested
    in steering the vehicles from their initial position to a prefixed
    end state, $\{(\mathbf{x}^*(t),\mathbf{v}^*(t)) \in
    \mathbb{R}^{Nm\times Nm} \,:\, \dot{\mathbf{x}}^*(t) =
    \mathbf{v}^*(t),$ for all $t \geq T_{\text{final}}\}$, while
    maintaining fixed the geometric configuration between them. Of the
    many available approaches for this task, we focus on the class of
    \textit{consensus algorithms} for systems modeled by a second-order
    dynamics in which the communication among the various vehicles is
    described in terms of the Laplacian of the graph of their
    connections. Specifically, we consider the following consensus model
    from~\cite{renconsensus}:
    \begin{equation}
    \left\lbrace\begin{array}{rl}
    \dot{\mathbf{x}}_i = & \mathbf{v}_i, \\
    \dot{\mathbf{v}}_i = & \ddot{\mathbf{x}}^*_i - \beta(\mathbf{x}_i - \mathbf{x}_i^*) - \gamma \beta (\mathbf{v}_i - \dot{\mathbf{x}}_i^*) \\
    & - \sum_{j=1}^{N} L_{i,j} [(\mathbf{x}_i - \mathbf{x}_i^*)-(\mathbf{x}_j - \mathbf{x}_j^*)] \\
    & - \gamma\sum_{j=1}^{N} L_{i,j} [(\mathbf{v}_i - \dot{\mathbf{x}}_i^*) - (\mathbf{x}_j - \dot{\mathbf{x}}_j^*)],
    \end{array}\right.
    \label{eq:consensus_system}
    \end{equation}
    which can be expressed in matrix form as
    \begin{equation*}
    \displaystyle \ddt{{\begin{bmatrix}
            \tilde{\mathbf{x}}\\
            \tilde{\mathbf{v}}
            \end{bmatrix}}} = \left(\begin{bmatrix}
    O_{N \times N} & I_N \\
    -(\beta I_N + L) & -\gamma (\beta I_N + L)
    \end{bmatrix} \otimes I_m\right) {\begin{bmatrix}
        \tilde{\mathbf{x}}\\
        \tilde{\mathbf{v}}
        \end{bmatrix}}
    \end{equation*}
    where $\tilde{\mathbf{x}} = \mathbf{x}^* - \mathbf{x}$, and $\tilde{\mathbf{v}} = \mathbf{v}^* - \mathbf{v}$.
    From Theorem 3.3~\cite{renconsensus} we can extract the
    following limit result for~\eqref{eq:consensus_system}.
    \begin{theorem}
        Let $L$ be the Laplacian of the graph of the connections in~\eqref{eq:consensus_system}. Let $\mu_i$ be the $i$-th eigenvalue of
        $-L$. Then, $\mathbf{x}\rightarrow\mathbf{x}^*$,
        $\mathbf{v}\rightarrow\mathbf{v}^*$ if
        \begin{equation}\label{eq:ren_bound}
        \gamma >\max_i\sqrt{2}\left(|\nu_i|\cos(\frac{\pi}{2})-\tan^{-1}\left(-\frac{\operatorname{Re}(\nu_i)}{\operatorname{Im}(\nu_i)}\right)
        \right)^{-\frac{1}{2}},
        \end{equation}
        $\nu_i=-\beta+\mu_i$. \label{teo:Ren}
    \end{theorem}
    The model (\ref{eq:consensus_system}) is extended here by
    considering a fractional power of the graph Laplacian, i.e.,
    $L^{\alpha}$, $\alpha\in (0,1)$ instead of $L$. The convergence
    analysis can be performed with the same tools used
    in~\cite{renconsensus} and gives a result analogous to Theorem
    \ref{teo:Ren} but with $\mu_i$ the eigenvalues of~$-L^{\alpha}$. The
    notable differences are that now the dynamics is faster as $\alpha$
    approaches $0$, together with the fact that increasing the amount of
    communication helps the vehicles in maintaining their formation; see
    the numerical experiment in Figure \ref{fig:fractional-consensus} in
    which it can be observed that the position at each time step of the
    vehicles resembles the initial one faster as $\alpha$ is smaller.
    \begin{figure}[htbp]
        \centering
        \subfloat[Trace of the position of the vehicles]{
            \includegraphics[width=0.7\columnwidth]{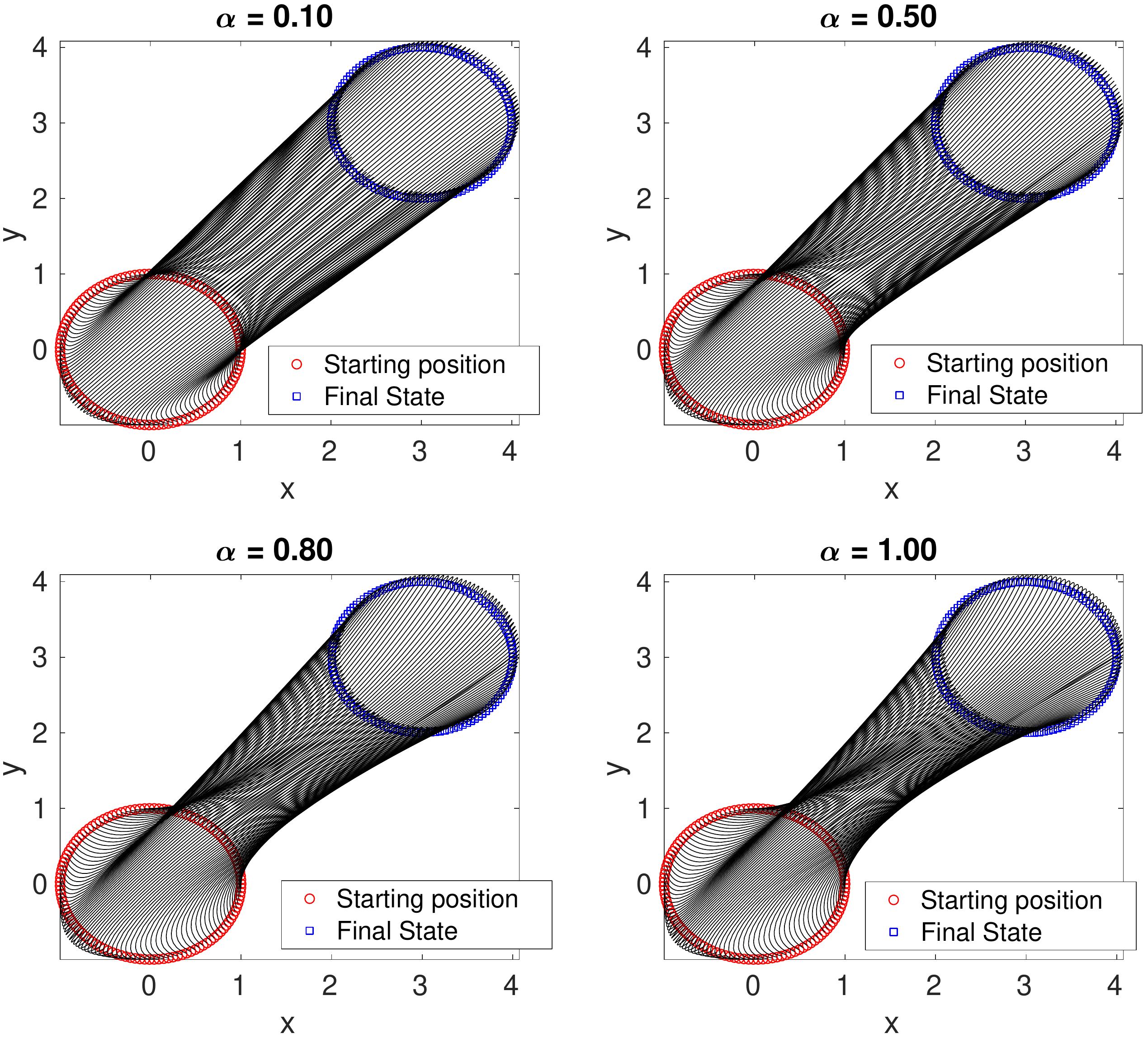}}

        \subfloat[Absolute error on the final position with respect to time]{
            \definecolor{mycolor1}{rgb}{0.00000,0.44700,0.74100}%
            \definecolor{mycolor2}{rgb}{0.85000,0.32500,0.09800}%
            \definecolor{mycolor3}{rgb}{0.92900,0.69400,0.12500}%
            \definecolor{mycolor4}{rgb}{0.49400,0.18400,0.55600}%
            \begin{tikzpicture}[scale=\columnwidth/38cm,samples=200]

            \begin{axis}[%
            width=8.218in,
            height=4.389in,
            at={(1.379in,0.592in)},
            scale only axis,
            xmin=0,
            xmax=5,
            xlabel style={font=\color{white!15!black}},
            xlabel={\Large $T (s)$},
            ymode=log,
            ymin=0.573030645285253,
            ymax=100,
            yminorticks=true,
            ylabel style={font=\color{white!15!black}},
            ylabel={\Large $\sqrt{ \|\mathbf{x}(t) - \mathbf{x}^*(t)\|_2^2 + \|\mathbf{y}(t) - \mathbf{y}^*(t)\|_2^2}$},
            axis background/.style={fill=white},
            axis x line*=bottom,
            axis y line*=left,
            legend style={legend cell align=left, align=left, draw=white!15!black}
            ]
            \addplot [color=mycolor1, line width=2.0pt]
            table[row sep=crcr]{%
                0   46.475800154489\\
                0.000654613870448322    46.4757500773167\\
                0.00130922774089664 46.4756914240857\\
                0.00196384161134497 46.4756239913782\\
                0.00638163148940162 46.4749512116563\\
                0.0107994213674583  46.473883627361\\
                0.0152172112455149  46.4724169023075\\
                0.0196350011235716  46.4705498073872\\
                0.0379173531092477  46.4585942824306\\
                0.0561997050949238  46.4398268834651\\
                0.0744820570806 46.4142897156923\\
                0.0927644090662761  46.3820400216082\\
                0.145813859956127   46.2510595987111\\
                0.198863310845978   46.0660370679422\\
                0.251912761735829   45.8292754146361\\
                0.30496221262568    45.5432937780285\\
                0.358011663515531   45.2107372590872\\
                0.43866171595006    44.6220554206336\\
                0.51931176838459    43.9416074880593\\
                0.59996182081912    43.1788312211283\\
                0.68061187325365    42.342857110758\\
                0.76126192568818    41.4424083645298\\
                0.84191197812271    40.4857444458072\\
                1.01771247170016    38.2425135852658\\
                1.19351296527761    35.8409844195533\\
                1.36931345885505    33.3431256609768\\
                1.5451139524325 30.8013680162861\\
                1.72091444600995    28.2595730619498\\
                1.8967149395874 25.7539096769399\\
                2.15956900512804    22.1373590514443\\
                2.42242307066868    18.7395888843954\\
                2.68527713620932    15.6091832769833\\
                2.94813120174996    12.7749522802642\\
                3.2109852672906 10.2499363566272\\
                3.47383933283124    8.03475917358868\\
                3.88880206619969    5.14629120175103\\
                4.30376479956813    2.94167507654172\\
                4.64663983312214    1.57571783375613\\
                5   0.573030645285253\\
            };
            \addlegendentry{\Large $\text{Error on the position }\alpha\text{ = 0.1}$}

            \addplot [color=mycolor2, line width=2.0pt]
            table[row sep=crcr]{%
                0   46.475800154489\\
                0.000627756506127168    46.475752440117\\
                0.00125551301225434 46.4756968414439\\
                0.0018832695183815  46.4756331735178\\
                0.00623892544637147 46.4749805420467\\
                0.0105945813743614  46.473944962797\\
                0.0149502373023514  46.4725230009891\\
                0.0193058932303414  46.4707141380682\\
                0.0403911525539567  46.4565490346738\\
                0.061476411877572   46.4334606893598\\
                0.0825616712011873  46.4015642284886\\
                0.103646930524803   46.3609888135621\\
                0.166889704739794   46.1885207820568\\
                0.230132478954786   45.9429434816168\\
                0.293375253169777   45.6282959137973\\
                0.356618027384768   45.248714658653\\
                0.41986080159976    44.8083755132761\\
                0.537206606118763   43.8443723027022\\
                0.654552410637766   42.7118222889199\\
                0.771898215156769   41.4362018650838\\
                0.889244019675771   40.0415275038227\\
                1.00658982419477    38.5501995286418\\
                1.22486071618948    35.5867142564407\\
                1.44313160818419    32.4763923169986\\
                1.66140250017889    29.3176338907817\\
                1.8796733921736 26.1888670563198\\
                2.0979442841683 23.1512680616304\\
                2.36346612248154    19.6436973345802\\
                2.62898796079478    16.3983094521859\\
                2.89450979910801    13.4517070396218\\
                3.16003163742125    10.8214202978633\\
                3.42555347573449    8.51070385657057\\
                3.7391672946488 6.18209201067874\\
                4.05278111356311    4.26468066981834\\
                4.36639493247742    2.72595975429489\\
                4.68000875139172    1.53382078058465\\
                5   0.67809979564361\\
            };
            \addlegendentry{\Large $\text{Error on the position }\alpha\text{ = 0.5}$}

            \addplot [color=mycolor3, line width=2.0pt]
            table[row sep=crcr]{%
                0   46.475800154489\\
                0.000714649093774652    46.4757447031449\\
                0.0014292981875493  46.4756790364116\\
                0.00214394728132395 46.4756029163432\\
                0.00721373716462588 46.4747780228043\\
                0.0122835270479278  46.4734350807512\\
                0.0173533169312297  46.4715701846416\\
                0.0224231068145316  46.4691833157998\\
                0.056572167432358   46.4397853218161\\
                0.0907212280501844  46.387434172073\\
                0.124870288668011   46.3127217471535\\
                0.159019349285837   46.2162918055291\\
                0.243666568958654   45.8876170489755\\
                0.328313788631471   45.4393952738432\\
                0.412961008304288   44.8817743876398\\
                0.497608227977105   44.2246769807777\\
                0.582255447649921   43.4777299382888\\
                0.771360293033488   41.5301036211417\\
                0.960465138417054   39.2769177578866\\
                1.14956998380062    36.8062564684096\\
                1.33867482918419    34.1948201977542\\
                1.52777967456775    31.5084832530877\\
                1.80890076669919    27.4903892359712\\
                2.09002185883063    23.5746484254117\\
                2.37114295096207    19.8693499809731\\
                2.6522640430935 16.4465995704223\\
                2.93338513522494    13.3501666805136\\
                3.23823735267326    10.3853335403021\\
                3.54308957012158    7.8352035151016\\
                3.8479417875699 5.68804897235597\\
                4.15279400501822    3.92034324924219\\
                4.45764622246654    2.50330110322359\\
                4.78473114346607    1.34686820380582\\
                5   0.803253023073446\\
            };
            \addlegendentry{\Large $\text{Error on the position }\alpha\text{ = 0.8}$}

            \addplot [color=mycolor4, line width=2.0pt]
            table[row sep=crcr]{%
                0   46.475800154489\\
                0.00076937829164787 46.4757396887256\\
                0.00153875658329574 46.4756673846207\\
                0.00230813487494361 46.4755829673977\\
                0.00778067239445547 46.4746508175328\\
                0.0132532099139673  46.4731155575564\\
                0.0187257474334792  46.4709730068776\\
                0.0241982849529911  46.4682235147394\\
                0.0538386965696841  46.4429832574697\\
                0.0834791081863772  46.4004565374942\\
                0.11311951980307    46.3410533545461\\
                0.142759931419763   46.2652124255882\\
                0.203180765565973   46.0614397930451\\
                0.263601599712183   45.7948634165463\\
                0.324022433858394   45.469269266026\\
                0.384443268004604   45.0883655396841\\
                0.444864102150814   44.6557728037324\\
                0.583337809486558   43.4879544980976\\
                0.721811516822302   42.1079227260637\\
                0.860285224158046   40.5541863721637\\
                0.99875893149379    38.8620742924824\\
                1.13723263882953    37.0637602592551\\
                1.38830862538639    33.6215626881678\\
                1.63938461194325    30.0726921847438\\
                1.89046059850011    26.5378285985614\\
                2.14153658505697    23.1088611459173\\
                2.39261257161382    19.8535857227633\\
                2.70380040925523    16.1278606959173\\
                3.01498824689665    12.7983349319585\\
                3.32617608453806    9.89108700520321\\
                3.63736392217947    7.40749251736524\\
                3.94855175982088    5.33241481333908\\
                4.27062646960939    3.58857088889684\\
                4.52641860850393    2.47760613099706\\
                4.78221074739846    1.59938033473864\\
                5   1.04688260563525\\
            };
            \addlegendentry{\Large $\text{Error on the position }\alpha\text{ = 1.0}$}

            \end{axis}

            \begin{axis}[%
            width=10.604in,
            height=5.385in,
            at={(0in,0in)},
            scale only axis,
            xmin=0,
            xmax=1,
            ymin=0,
            ymax=1,
            axis line style={draw=none},
            ticks=none,
            axis x line*=bottom,
            axis y line*=left,
            legend style={legend cell align=left, align=left, draw=white!15!black}
            ]
            \end{axis}
            \end{tikzpicture}%
        }
        \caption{(Color online) We consider here the test cases in which
            $n=120$ vehicles are uniformly distributed on the unit circle
            $\mathbf{x} = (\cos(t_i),\sin(t_i))_i$, $\{t_i = {2\pi
                i}{n}\}_{i=1}^{n}$, with starting velocity given $\mathbf{v} =
            \dot{\mathbf{x}}$, i.e., they are following a uniform circular
            motion. The desired ending state is represented again by a circle of
            unit radius and uniformly distributed vehicles but with center in
            $(3,3)$ and zero terminal velocity. The parameter $\beta$ is $0.5$
            while $\gamma$ is computed as the lower bound
            in~\eqref{eq:ren_bound} plus one. The communication graph between
            the vehicles is the directed cycle.}\label{fig:fractional-consensus}
    \end{figure}
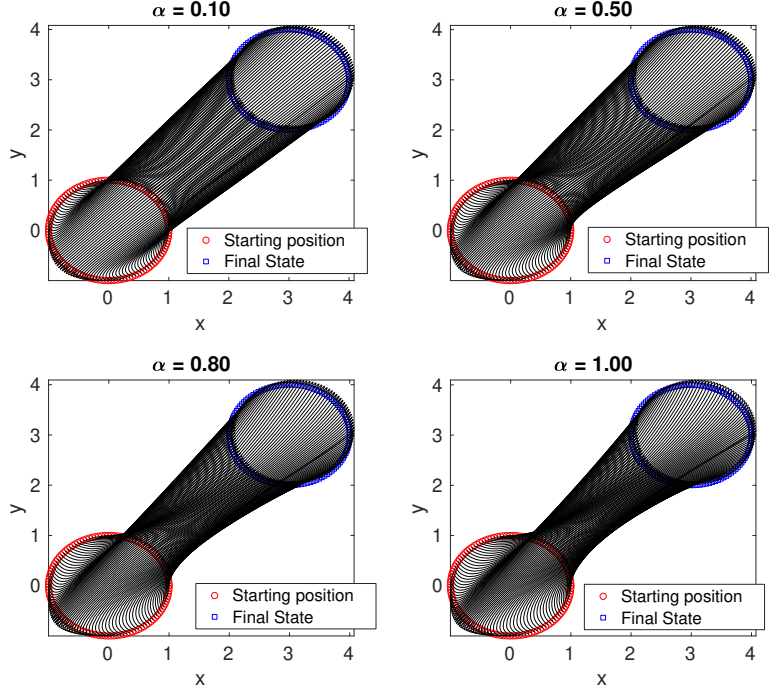
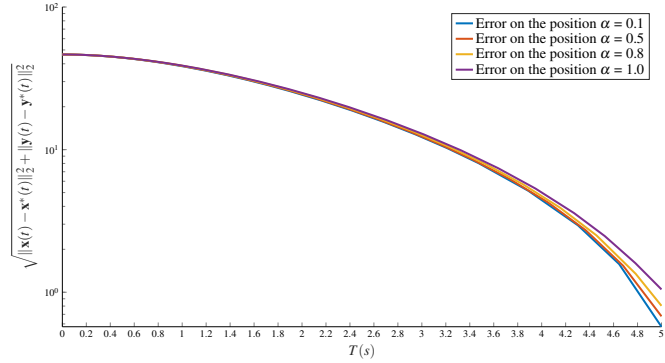

    \section{Conclusions}

    In this paper we have investigated nonlocal diffusion dynamics (both discrete and
    continuous in time) on undirected as well as on directed networks using fractional powers of a suitable %
    version of the graph Laplacian and its normalized counterpart. In order to treat the directed case, we have discussed the definition of the $\alpha$th power
    of a nonsymmetric graph Laplacian. We proved also that the proposed dynamic exhibits a \emph{superdiffusive} behavior
    for both the undirected and directed path graph thus strengthening the analogy with the continuous fractional Laplacian.
    We have obtained analytical solutions for two simple
    directed graphs (a periodic one and an absorbing one) and highlighted some differences
    and similarities with fractional diffusion on related undirected graphs.    Experiments on a
    few real world examples indicate that, similar to the undirected case, nonlocal (fractional) diffusion
    and related random walks on directed graphs result in more efficient navigation of complex
    directed networks than using the standard (local) counterparts.  Finally, we have extended
    an existing consensus models for vehicle motions on directed networks to one driven by a
    fractional nonsymmetric Laplacian and observed that the system displays faster convergence
    to consensus than the standard (nonfractional) model.

    In conclusion, the dynamics of nonlocal fractional diffusion appears to be a useful tool in the study
    of several problems involving directed as well as undirected graph models.

    \section*{Funding}
    This work was supported in part by the Tor Vergata University \lq\lq
    Beyond Borders\rq\rq\ program through the project ASTRID, CUP
    E84I19002250005; by the INdAM--GNCS projects \lq\lq Tecniche
    innovative e parallele per sistemi lineari e non lineari di grandi
    dimensioni, funzioni ed equazioni matriciali ed
    applicazioni\rq\rq, \lq\lq Nonlocal models for the analysis of
    complex networks\rq\rq and \lq\lq Metodi low-rank per problemi di algebra lineare con
struttura data-sparse\rq\rq.

    \section*{Acknowledgment}
  {We would like to thank two anonymous referees for their helpful comments on an earlier draft
  of the paper.}

\bibliographystyle{comnet}
\bibliography{nonsymfractlap.bib}

\begin{thebibliography}{00}

\bibitem{Abadias}
Abadias, L., De~Le\'{o}n-Contreras, M. {\&} Torrea, J.~L. (2017)  {Non-local
  fractional derivatives. {D}iscrete and continuous}. {\em J. Math. Anal.
  Appl.}, \textbf{449}(1), 734--755.

\bibitem{MR2915277}
Bauer, F. (2012)  Normalized graph {L}aplacians for directed graphs. {\em
  Linear Algebra Appl.}, \textbf{436}(11), 4193--4222.

\bibitem{BB14}
Benzi, M. {\&} Boito, P. (2014)  Decay properties for functions of matrices
  over {$C^*$}-algebras. {\em Linear Algebra Appl.}, \textbf{456}, 174--198.

\bibitem{MR2455657}
Benzi, M. {\&} Razouk, N. (2007/08)  Decay bounds and {$O(n)$} algorithms for
  approximating functions of sparse matrices. {\em Electron. Trans. Numer.
  Anal.}, \textbf{28}, 16--39.

\bibitem{MR3391978}
Benzi, M. {\&} Simoncini, V. (2015)  Decay bounds for functions of {H}ermitian
  matrices with banded or {K}ronecker structure. {\em SIAM J. Matrix Anal.
  Appl.}, \textbf{36}(3), 1263--1282.

\bibitem{MR1298430}
Berman, A. {\&} Plemmons, R.~J. (1994) {\em Nonnegative Matrices in the
  Mathematical Sciences}, volume~9 of {\em Classics in Applied Mathematics}.
Society for Industrial and Applied Mathematics (SIAM), Philadelphia, PA.
Revised reprint of the 1979 original.

\bibitem{Chapman}
Chapman, S. (1911)  {On {N}on-{I}ntegral {O}rders of {S}ummability of {S}eries
  and {I}ntegrals}. {\em Proc. London Math. Soc. (2)}, \textbf{9}, 369--409.

\bibitem{MR2135772}
Chung, F. (2005)  Laplacians and the {C}heeger inequality for directed graphs.
  {\em Ann. Comb.}, \textbf{9}(1), 1--19.

\bibitem{Crouzeix04}
Crouzeix, M. (2004)  Bounds for analytical functions of matrices. {\em Integral
  Equations Operator Theory}, \textbf{48}, 461--477.

\bibitem{Crouzeix07}
Crouzeix, M. (2007)  Numerical range and functional calculus in {H}ilbert
  space. {\em J. Funct. Anal.}, \textbf{244}(2), 668--690.

\bibitem{estradamultihopper}
Estrada, E., Delvenne, J.-C., Hatano, N., Mateos, J.~L., Metzler, R., Riascos,
  A.~P. {\&} Schaub, M.~T. (2018a)  Random multi-hopper model: super-fast
  random walks on graphs. {\em J. Complex Netw.}, \textbf{6}(3), 382--403.

\bibitem{Estrada2017307}
Estrada, E., Hameed, E., Hatano, N. {\&} Langer, M. (2017)  Path Laplacian
  operators and superdiffusive processes on graphs. I. One-dimensional case.
  {\em Linear Algebra and Its Applications}, \textbf{523}, 307--334.

\bibitem{Estrada2018373}
Estrada, E., Hameed, E., Langer, M. {\&} Puchalska, A. (2018b)  Path Laplacian
  operators and superdiffusive processes on graphs. II. Two-dimensional
  lattice. {\em Linear Algebra and Its Applications}, \textbf{555}, 373--397.

\bibitem{MR700883}
Fiedler, M. {\&} Schneider, H. (1983)  Analytic functions of {$M$}-matrices and
  generalizations. {\em Linear and Multilinear Algebra}, \textbf{13}(3),
  185--201.

\bibitem{MR2599831}
Guo, C.-H. (2010)  On {N}ewton's method and {H}alley's method for the principal
  {$p$}th root of a matrix. {\em Linear Algebra Appl.}, \textbf{432}(8),
  1905--1922.

\bibitem{MR2396439}
Higham, N.~J. (2008) {\em Functions of Matrices. Theory and Computation}.
Society for Industrial and Applied Mathematics (SIAM), Philadelphia, PA.

\bibitem{MR2978290}
Horn, R.~A. {\&} Johnson, C.~R. (2013) {\em Matrix Analysis}.
Cambridge University Press, Cambridge, second edition.

\bibitem{ilic2005numerical}
Ilic, M., Liu, F., Turner, I. {\&} Anh, V. (2005)  Numerical approximation of a
  fractional-in-space diffusion equation. {I}. {\em Fract. Calc. Appl. Anal.},
  \textbf{8}(3), 323--341.

\bibitem{ilic2006numerical}
Ilic, M., Liu, F., Turner, I. {\&} Anh, V. (2006)  Numerical approximation of a
  fractional-in-space diffusion equation. {II}. {W}ith nonhomogeneous boundary
  conditions. {\em Fract. Calc. Appl. Anal.}, \textbf{9}(4), 333--349.

\bibitem{MR1812534}
Iserles, A. (2000)  How large is the exponential of a banded matrix?. {\em New
  Zealand J. Math.}, \textbf{29}(2), 177--192.
Dedicated to John Butcher.

\bibitem{MR0115196}
Kemeny, J.~G. {\&} Snell, J.~L. (1960) {\em Finite {M}arkov Chains}.
The University Series in Undergraduate Mathematics. D. Van Nostrand Co., Inc.,
  Princeton, N.J.-Toronto-London-New York.

\bibitem{Kuttner}
Kuttner, B. (1957)  {On differences of fractional order}. {\em Proc. London
  Math. Soc. (3)}, \textbf{7}, 453--466.

\bibitem{LizamaRoncal}
Lizama, C. {\&} Roncal, L. (2018)  H\"{o}lder-{L}ebesgue regularity and almost
  periodicity for semidiscrete equations with a fractional {L}aplacian. {\em
  Discrete Contin. Dyn. Syst.}, \textbf{38}(3), 1365--1403.

\bibitem{approx}
Meinardus, G. (1967) {\em Approximation of Functions: Theory and Numerical
  Methods}.
Springer, Berlin.

\bibitem{randomwalkfractional}
Metzler, R. {\&} Klafter, J. (2000)  The random walk's guide to anomalous
  diffusion: a fractional dynamics approach. {\em Phys. Rep.}, \textbf{339}(1),
  77.

\bibitem{page1999pagerank}
Page, L., Brin, S., Motwani, R. {\&} Winograd, T. (1999)  The PageRank citation
  ranking: Bringing order to the web.. Technical report, Stanford InfoLab.

\bibitem{renconsensus}
Ren, W. (2007)  Consensus strategies for cooperative control of vehicle
  formations. {\em IET Control Theory \& Applications}, \textbf{1},
  505--512(7).

\bibitem{PhysRevE90032809}
Riascos, A.~P. {\&} Mateos, J.~L. (2014)  Fractional dynamics on networks:
  {E}mergence of anomalous diffusion and {L}\'evy flights. {\em Phys. Rev. E},
  \textbf{90}, 032809.

\end{thebibliography}

\end{document}